\pdfoutput=1
\documentclass[journal]{IEEEtran}

\usepackage{multirow}
\usepackage{diagbox}
\usepackage{amssymb}
\usepackage{amsmath}
\usepackage{graphicx}
\usepackage{cite}
\usepackage{citesort}
\usepackage{booktabs}
\usepackage{subfigure}
\usepackage{graphicx,epstopdf}
\usepackage{epsfig}	
\usepackage{cite,graphicx,amsmath,amssymb}
\usepackage{comment}
\usepackage{amssymb}
\usepackage{amsmath}
\usepackage{cite}
\usepackage{url}
\usepackage{xcolor}
\usepackage{cite,graphicx,amsmath,amssymb}
\usepackage{subfigure}
\usepackage{citesort}
\usepackage{fancyhdr}
\usepackage{mdwmath}
\usepackage{mdwtab}
\usepackage{caption}
\usepackage{amsthm}
\usepackage{lipsum}
\usepackage{array}
\usepackage{booktabs}
\usepackage{makecell}
\newtheorem{theorem}{Theorem}

\newtheorem{lemma}{Lemma}

\newtheorem{corollary}{Corollary}

\newtheorem{remark}{Remark}

\makeatletter
\def\ScaleIfNeeded{
\ifdim\Gin@nat@width>\linewidth \linewidth \else \Gin@nat@width
\fi } \makeatother

\begin{document}

\title{\Huge{Secrecy Outage Probability Analysis for Downlink RIS-NOMA Networks with On-Off Control}}

\author{ Yingjie~Pei, Xinwei\ Yue,~\IEEEmembership{Senior Member,~IEEE}, Wenqiang\ Yi,~\IEEEmembership{Member,~IEEE}, Yuanwei\ Liu,~\IEEEmembership{Senior Member,~IEEE}, Xuehua\ Li,~\IEEEmembership{Member,~IEEE}, and Zhiguo\ Ding,~\IEEEmembership{Fellow,~IEEE}

\thanks{This work was supported in part by National Natural Science Foundation of China (Grant 62071052) and in part by Beijing Natural Science Foundation (Grant L222004). \emph{(Corresponding author: Xinwei Yue.)}}
\thanks{Y. Pei, X. Yue and X. Li are with the Key Laboratory of Information and Communication Systems, Ministry of Information Industry and also with the Key Laboratory of Modern Measurement $\&$ Control Technology, Ministry of Education, Beijing Information Science and Technology University, Beijing 100101, China (email: \{yingjie.pei, xinwei.yue and lixuehua\}@bistu.edu.cn).}
\thanks{W. Yi and Y. Liu are with the School of Electronic Engineering and Computer Science, Queen Mary University of London, London E1 4NS, U.K. (e-mail: w.yi, yuanwei.liu@qmul.ac.uk).}
\thanks{Z. Ding is with the Department of Electrical Engineering, Princeton University, Princeton, USA and also with the School of Electrical and Electronic Engineering, the University of Manchester, Manchester, U.K. (email: zhiguo.ding@manchester.ac.uk).}
}
\maketitle

\begin{abstract}
Reconfigurable intelligent surface (RIS) has been regarded as a promising technology since it has ability to create the favorable channel conditions. This paper investigates the secure communications of RIS assisted non-orthogonal multiple access (NOMA) networks, where both external and internal eavesdropping scenarios are taken into consideration.
More specifically, novel approximate and asymptotic expressions of secrecy outage probability (SOP) for the \emph{k}-th legitimate user (LU) are derived by invoking imperfect successive interference cancellation (ipSIC) and perfect successive interference cancellation (pSIC).
To characterize the secrecy performance of RIS-NOMA networks, the diversity order of the \emph{k}-th LU with ipSIC/pSIC is obtained in the high signal-to-noise ratio region. The secrecy system throughput of RIS-NOMA networks is discussed in delay-limited transmission mode. Numerical results are presented to verify theoretical analysis that: i) The SOP of RIS-NOMA networks is superior to  that of RIS assisted orthogonal multiple access (OMA) and conventional cooperative communication schemes; ii) As the number of reflecting elements increases, the RIS-NOMA networks are capable of achieving the enhanced secrecy performance; and iii) The RIS-NOMA networks have better secrecy system throughput than that of RIS-OMA networks and conventional cooperative communication schemes.
\end{abstract}
\section{Introduction}
As one of pivotal technologies for the next generation communication networks, non-orthogonal multiple access (NOMA) has attracted extensive attention due to its remarkable spectral efficiency and superiority for massive connectivity compared to orthogonal multiple access (OMA)\cite{2015LinglongNOMA,2017YuanweiNOMA,you2021towards}. More specifically, NOMA enables multiple users to transmit information in the same resource block through different power allocations \cite{2014ChoiNOMA,2017ZhiguoNOMA}. By taking into account the choices of users' rates, the authors of \cite{ding2014performance} investigated the performance of NOMA networks in terms of outage probability and ergodic rate. Futhermore, the cooperative NOMA scheme was surveyed in \cite{2018XinweiFDNOMA}, where the nearby user with better channel conditions was regarded as the relaying. To satisfy the needs of cell-edged users, the successive user relaying in NOMA networks was considered in \cite{2019SuccessUserReplay}, where the wireless links of distant user were established with the help of user relaying. As a further development, the authors of \cite{2019RelaySelCoNOMA} investigated the relaying selection and interference elimination schemes to enhance the outage behaviors. In \cite{2020XingwangRHINOMA}, the impact of residual hardware impairments on NOMA networks was surveyed which leads to zero diversity order in high signal-to-noise ratio (SNR) region. Regarding to the massive connectivity scenarios, the authors of \cite{2020MassConnNOMA} analysed the achievable rate of NOMA networks according to the minimal pairing distance of users.

Since the wireless signals are easily exposed to complex electromagnetic environments, the physical layer security technology has ability to ensure the data secure transmission. \cite{2015YangnanPLS,2018yongpengPLS,2018LvluNOMAPLS}. Until now, the secure communications of NOMA networks have been discussed in many treatises. In \cite{2017YuanweiANNOMAPLS}, the authors evaluated the secrecy performance of large-scale downlink NOMA networks with the presence of an eavesdropper (Eve). A millimeter wave based NOMA secrecy beamforming scheme was proposed in \cite{2020XiaomingNOMAPLS}, where the secrecy outage probability (SOP) of paired users was obtained. Given the practical distribution of users, the authors of \cite{2021CaihongANNOMAPLS} studied the secrecy NOMA transmission with randomly distributed legitimate users (LUs) and the SNR was assumed to be fixed/dynamic. For further security improvement, the authors in \cite{2021ZhiguoJammingNOMAPLS} studied the SOP of NOMA networks by invoking the effective full-duplex (FD) jamming scheme. Different from the point-to-point orthogonal transmission, the superimposed signal is transmitted simultaneously to all users in NOMA networks, which makes it possible for users with poor channel condition to wiretap others' information as internal Eves \cite{2017ZhiguoInterEve}. Triggered by this, the authors of \cite{2020XinweiUnifiedNOMA} researched the secrecy performance of unified NOMA framework by taking into account stochastic geometry under external and internal eavesdropping scenarios. Additionally, the secure communications of cooperative NOMA networks were revealed in \cite{2022VincentNOMAPLS}, where the secrecy throughput was improved greatly.

Reconfigurable intelligent surface (RIS) has been widely considered as an effective approach to improve the reliability of wireless communication networks \cite{2019QingqingIRS,2020QingqingIRSmagzine,2021QingqingIRS}. With the assistance of low-cost reconfigurable passive elements \cite{2021CunhuaIRSmagzine}, RIS is capable of reconfiguring the amplitude and phase shifts of incident electromagnetic signals, via adjusting the propagation direction of reflecting signals. Compared with traditional active relaying, the authors of \cite{2020ZhangruiIRS} highlighted the achievable rate of RIS-enhanced networks. In \cite{2021BjornIRS}, the ergodic capacity of RIS assisted single-input single-output system were analysed in detail. For further exploration, the authors of \cite{2020TaoqinIRS} considered extra direct links between transceivers and surveyed outage behaviors of single-antenna equipped communication networks by the aid of reasonable RIS deployment. In \cite{2020SpeEneEffIRS}, the energy efficiency of RIS-assisted multiple-input single-output system was highlighted carefully. Limited by fixed location and finite coverage, however, single RIS can be difficult to serve multiple users within a wide range. To resolve this issue, multiple RISs were introduced in \cite{2020ZhangruiIRSSpatThroput} to realize enhanced spatial throughput of a single-cell multiuser system. Furthermore, the authors of \cite{2021MultiIRS_OP} studies the strategy of selecting multiple RISs based on the optimal SNR.

In light of the above discussions, researchers have dived into surveying the integrate of RIS to NOMA networks \cite{2020QingqingIRSNOMA,2020YuanweiIRSNOMA,2021YuanweiIRSNOMA}. Several excellent works have manifested that the performance of RIS-NOMA scheme surpasses that of conventional NOMA transmission. The authors of \cite{2021TahirIRSNOMA} took direct uplinks into account and verified the non-line-of-sight severed users can distinctly profit from RIS compared to the line-of-sight severed users. In \cite{2021YuanweiDownUplinkIRSNOMA}, the authors synthetically investigated uplink and downlink RIS-NOMA scenarios, where the closed-form expression of outage probability for the cell-edge user was derived. However, these works are supposed that the incident and reflecting waves at RIS can perfectly realize coherent matching with the assistance of RIS phase shifters, which will lead to additional signaling overhead. To tackle this problem, the coherent phase shifting and random phase shifting of RIS-NOMA were proposed in \cite{2020ZhiguoTwoPhaseShift}, which achieves the tradeoff between complexity and reliability. Furthermore, the authors of \cite{2020ZhiguoSimpleDesignIRSNOMA} applied the on-off control scheme as a special case of random phase shifting in NOMA networks to enhance the outage behaviors of RIS-NOMA networks.
The ergodic rate of RIS-NOMA for multiple users were evaluated in-depth with imperfect/perfect successive interference cancellation (ipSIC/pSIC) in \cite{2022XinweiIRSNOMA}, where the 1-bit scheme was utilized at RIS. Further, the authors of \cite{2022ZeyuIRSNOMA} analysed the required power and outage performance by introducing continuous/discrete phase shifting in RIS-NOMA networks with multiple antennas.
The authors of \cite{2021YuanweiMultiIRSNOMA} proposed a multiple RISs-assisted NOMA scheme and each RIS was assumed to serve a targeted user precisely.

RIS is able to be explored for ensuring the secure communication of wireless networks and blended with physical layer security. A RIS assisted secure wireless transmission system was considered in \cite{2019CuimiaoRISsec}, where the secrecy rate was maximized to protect the privacy of LUs. In \cite{2019XianghaoRIS}, the authors evaluated secrecy performance of RIS assisted multiple-input single-output system by designing phase shift and beamforming. The authors of \cite{2020MIMOIRSsec_2} jointly optimized the transmit and phase shift coefficients to obtain the maximal secrecy rate in RIS-assisted multiple-input multiple-output (MIMO) system. Given the passivity of Eves, a robust RIS-aided transmission scheme was proposed in \cite{2020XianghaoRobust}, where the channel state information (CSI) of Eves cannot be perfectly acquired. To further explore the secrecy performance of RIS-NOMA networks, the authors of \cite{2021QingqingNOMAIRSPLS} focused on defending lacking CSI eavesdropping by employing artificial noise (AN). In \cite{2022YuanweiIRSNOMAPLS}, RIS was utilized to guarantee secure NOMA transmission with only Eve's statistical CSI. In \cite{2022MaxminIRSNOMAPLS}, the active and passive beamforming vectors were redesigned to improve the security of RIS-NOMA while ensuring users' fairness. The authors in \cite{2020IRSNOMAPLS} deployed RIS to achieve higher channel gains and protect users from vicious external Eves in NOMA networks. On this basis, further investigation was processed in \cite{2022CaihongIRSNOMAPLS} to analyse the negative effect caused by external/internal Eve on the SOP as well as effective secrecy throughput of RIS-NOMA networks. Recently, the secrecy performance of simultaneously transmitting and reflecting (STAR) RIS aided NOMA networks was investigated in \cite{2022STARRISPLSAN}, with the unrealistic assumption of perfect eavesdropping CSI. Furthermore, the authors of \cite{2022STARRISPLSuplink} surveyed the secrecy uplink transmission of STAR-RIS assisted NOMA networks by jointly optimizing receive beamforming, transmitting power, and reflecting/transmission coefficients, where the instantaneous CSI of Eve is also required.

\subsection{Motivations and Contributions}
The previously mentioned theoretical literature has laid a solid foundation for the comprehension of RIS-NOMA networks. With regard to urban secure communication, the authors of \cite{2020XinweiUnifiedNOMA} highlighted the secrecy outage behaviour in an unified NOMA network without utilizing RIS. Additionally, a secure robust design of RIS-NOMA networks was proposed in \cite{2021QingqingNOMAIRSPLS}, with the impractical assumption of continuous phase shift and pSIC. Inspired by these treatises, we specifically investigate secure communications of RIS-NOMA networks, where the superimposed wireless signals are delivered from base station (BS) to multiple non-orthogonal users via the assistance of RIS with the presence of Eves. More specifically, both external/internal eavesdropping scenarios are fully discussed and the residual interference caused by ipSIC is also taken into account. Given the coherent phase shifting may cause excessive signalling overhead and is impractical due to the finite resolution of phase shifters at RIS, the on-off control is selected as a feasible scheme to redesign the phase shifts of RIS for realizing the secure transmission of RIS-NOMA networks \cite{2020ZhiguoSimpleDesignIRSNOMA,cui2014coding}. Referring to the aforementioned explanations, the primary contributions of this manuscript can be summarized in detail as follows:

\begin{enumerate}
  \item We study the secure performance of RIS-NOMA networks, where both external and internal eavesdropping scenarios are taken into consideration. Specifically, we derive approximate and asymptotic expressions of SOP for the \emph{k}-th LU under both ipSIC/pSIC conditions.
 To glean more insights, the secrecy diversity orders of the \emph{k}-th LU at high SNRs are obtained. We observe that secrecy diversity order for the \emph{k}-th LU with pSIC is equal to \emph{k}, while that with ipSIC is equal to zero since there exists an error floor.
  \item We verify that the secrecy outage behaviours of RIS-NOMA networks outperform that of RIS-OMA and conventional cooperative relaying schemes, i.e., amplify-and-forward (AF) relaying and decode-and-forward (DF) relaying in both half-duplex (HD) and FD modes. With increasing the number of reflecting elements, RIS-NOMA can achieve superior secrecy performance. Moreover, the SOP rises remarkably by augmenting the values of residual interference, target secrecy rates and the distances from RIS to BS/LUs. We also notice that the valid power allocation in external eavesdropping scenarios can exert an opposite impact on the transmission privacy in internal eavesdropping scenarios.
  \item We analyse the secrecy system throughput of RIS-NOMA networks in delay-limited transmission mode under external and internal eavesdropping scenarios. We confirm that the secrecy system throughput converge to a consistent value in high SNR region. In addition, the system throughput of RIS-NOMA is superior to that of RIS-OMA, AF relaying and HD/FD DF relaying since the RIS assisted wireless communications are able to realize high spectral efficiency and reliable channel environment.
\end{enumerate}

\subsection{Organization and Notation}
The remainder of this paper is given in the following. In Section \ref{Section II}, the system model and signal formulas for RIS-NOMA secure communication networks are depicted. The channel statistics and the SOP analysis in both external/internal eavesdropping scenarios are introduced in Section \ref{Section_III}. Numerical results are present in Section \ref{SectionIV} and followed by our conclusion in Section \ref{SectionV}.

The primary notations appeared in this paper are shown as follows. $\mathbb{E}\left\{  \cdot  \right\}$ represents the expectation operation. ${F_X}\left(  \cdot  \right)$ and ${f_X}\left(  \cdot  \right)$ denote the cumulative distribution function (CDF) and the probability density function (PDF) of parameter \emph{X}, respectively. ${{\mathbf{I}}_\mu }$ denotes a $\mu  \times \mu $ identity matrix and ${{\mathbf{1}}_\mu }$ denotes a $\mu  \times 1$ all-ones column vector. $ \otimes $ means Kronecker product.
\section{System Model}\label{Section II}
\subsection{System Descriptions}
Considering a RIS-assisted NOMA secure communication scenario as depicted in Fig. 1, where a BS transmits the superposed signals to \emph{K} LUs via the assistance of a RIS in the presence of an Eve. For the purpose of intuitive analysis, we assume that each node in the system is equipped with single antenna. The RIS is mounted with \emph{M} reconfigurable reflecting elements, which can be controlled by a programmable logic device. The complex channel coefficient from the BS to RIS, from the RIS to \emph{k}-th LU, and from the RIS to the Eve are represented by ${{\mathbf{h}}_{br}} \in {\mathbb{C}^{M \times 1}}$, ${{\mathbf{h}}_{rk}} \in {\mathbb{C}^{M \times 1}}$ and ${{\mathbf{h}}_{re}} \in {\mathbb{C}^{M \times 1}}$, respectively. We suppose that the perfect CSI of LUs can be obtained by means of pilot signals \cite{zheng2020intelligent,zheng2019intelligent}. However, the CSI of Eve is hard to acquire since it barely exchange CSI with BS. Furthermore, even if the BS can get the eavesdropping CSI with little leaked information from Eve, it tends to be outdated \cite{2020MIMOIRSsec_2}. As a consequence, we assume the CSI of Eves is unknown to the BS\footnote{In practical cases, the channel estimation process always results in errors due to the RIS quantization errors \cite{2020ChangshengYouChanEsti} or noisy estimation \cite{2020LiangLiuChannEsti}. Most literature focuses on improving the robustness of secure RIS networks with the assistance of artificial noise and the joint optimization of transmitting beamformer as well as RIS phase matrix \cite{2020XianghaoRobust,2021QingqingNOMAIRSPLS}.}. Given that obstacles can scatter a great deal of radio signals in practical urban district scenarios, all wireless links from BS to RIS and to LUs/Eve are supposed to be Rayleigh fading channels. The direct links between BS and LUs/Eve experience severe attenuation, and thus the communication can be only established through the RIS. Without loss of generality, the cascade channel gains of BS-RIS and-LUs are ordered as ${\left| {{\mathbf{h}}_{_{r1}}^H{\mathbf{\Theta }}{{\mathbf{h}}_{br}}} \right|^2} \leqslant  \cdots  \leqslant {\left| {{\mathbf{h}}_{_{rk}}^H{\mathbf{\Theta }}{{\mathbf{h}}_{br}}} \right|^2} \leqslant  \cdots  \leqslant {\left| {{\mathbf{h}}_{_{rK}}^H{\mathbf{\Theta }}{{\mathbf{h}}_{br}}} \right|^2}$, where ${\mathbf{\Theta }} = \beta {\text{diag}}\left( {{e^{j{\theta _1}}},{e^{j{\theta _2}}}...,{e^{j{\theta _M}}}} \right) \in {\mathbb{C}^{M \times M}}$ represents the reflecting parameters matrix of RIS, $\beta  \in \left[ {0,1} \right]$ and ${\theta _m} \in \left[ {0,2\pi } \right)$ are the reflecting amplitude parameter and the phase shift of the \emph{m}-th reflecting element, respectively.

\begin{figure}[t!]
    \begin{center}
        \includegraphics[width=2.784in,  height=1.84in]{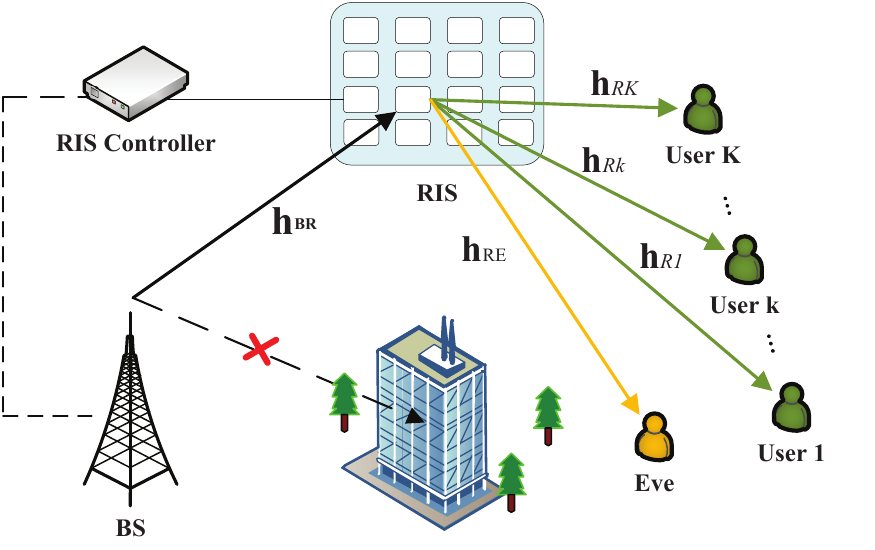}
        \caption{System model of RIS assisted NOMA secure communication networks.}
        \label{sys_model.eps}
    \end{center}
\end{figure}

\subsection{Signal Model}
In RIS-NOMA communication networks, the signal streams of LUs are broadcast to each LU by using the superposition coding scheme. Hence the received signal at \emph{k}-th LU is shown as
\begin{align}\label{the received signal at users}
{y_k} = {\mathbf{h}}_{rk}^H{\mathbf{\Theta }}{{\mathbf{h}}_{br}} {\sum\limits_{i = 1}^K {\sqrt {{a_i}{P_s}} {x_i}} }  + {n_k},
\end{align}
where ${{x_i}}$ is supposed to be unity power signal for \emph{i}-th LU, i.e., $\mathbb{E}\{ {\left| {{x_i}} \right|^2}\}  = 1$. The ${{P_s}}$ is the normalized transmission power at BS. To ensure user fairness, ${{a_i}}$ denotes the power allocation factor of \emph{i}-th LU, which meets the relationship ${a_1} \geqslant  \cdots  \geqslant {a_k} \geqslant  \cdots  \geqslant {a_K}$ with $\sum\limits_{i = 1}^K {{a_i}}  = 1$.
The ${{\mathbf{h}}_{br}} = {\left[ {h_{br}^1 \cdots h_{br}^m \cdots h_{br}^M} \right]^H}$, where the $h_{br}^m = \sqrt {d_{br}^{ - \alpha }} \tilde h_{br}^m \sim \mathcal{C}\mathcal{N}\left( {0,{N_{br}}} \right)$ denotes the complex channel coefficient between BS and the \emph{m}-th reflecting element of RIS and $\tilde h_{br}^m \sim \mathcal{C}\mathcal{N}\left( {0,1} \right)$. The ${{\mathbf{h}}_{rk}} = {\left[ {h_{rk}^1 \cdots h_{rk}^m \cdots h_{rk}^M} \right]^H}$, where the $h_{rk}^m = \sqrt {d_{rk}^{ - \alpha }} \tilde h_{rk}^m \sim \mathcal{C}\mathcal{N}\left( {0,{N_{rk}}} \right)$ denotes the complex channel coefficient from the \emph{m}-th reflecting element of RIS to the \emph{k}-th LU and $\tilde h_{rk}^m \sim \mathcal{C}\mathcal{N}\left( {0,1} \right)$. The $\alpha $ is path loss exponent. The ${d_{br}}$ and the ${d_{rk}}$ denote the distance from BS to RIS and from RIS to the \emph{k}-th LU, respectively. The ${n_k}$ denotes the additive white Gaussian noise (AWGN) with mean power parameter ${N_U}$.

According to the NOMA principle, the received signal-to-interference-plus-noise ratio (SINR) for the \emph{k}-th LU to detect the \emph{g}-th LU information (\emph{k} $\geqslant $ \emph{g}) can be given by
\begin{align}\label{SINR k decode g}
{\gamma _{k \to g}} =  \mathop {}\nolimits^{} \mathop {}\nolimits^{} \frac{{\rho {{\left| {{\mathbf{h}}_{rk}^H{\mathbf{\Theta }}{{\mathbf{h}}_{br}}} \right|}^2}{a_g}}}{{\rho {{\left| {{\mathbf{h}}_{rk}^H{\mathbf{\Theta }}{{\mathbf{h}}_{br}}} \right|}^2} {{\nu _g}}  + \varpi \rho {{\left| {{h_{ipu}}} \right|}^2} + 1}},
\end{align}
where ${\nu _g} = \sum\nolimits_{i = g + 1}^K {{a_i}}$, $\rho  = \frac{{{P_s}}}{{{N_U}}}$ denotes the transmit SNR of legitimate channels between the BS and LUs. The $\varpi  \in \left[ {0,1} \right]$ is the residual interference level for SIC. More specifically, $\varpi  = 0$ and $\varpi  \ne 0$ indicate the operations of pSIC and ipSIC, respectively. Without loss of generality, we assume that the residual interference caused by ipSIC is modeled as the Rayleigh fading and corresponding complex channel parameter is represented by ${h_{ipu}} \sim \mathcal{C}\mathcal{N}\left( {0,{N_{ipu}}} \right)$ \cite{2022XinweiIRSNOMA}.

After striking out the previous (\emph{K}-1) LUs' messages with SIC, the received SINR for the \emph{K}-th LU to detect its own information can be shown as
\begin{align}\label{SINR K decode K}
{\gamma _K} = \frac{{\rho {{\left| {{\mathbf{h}}_{rK}^H{\mathbf{\Theta }}{{\mathbf{h}}_{br}}} \right|}^2}{a_K}}}{{\varpi \rho {{\left| {{h_{ipu}}} \right|}^2} + 1}}.
\end{align}

The broadcast characteristics are capable of increasing the susceptibility to potential eavesdropping during the propagation of wireless signals. Inspired by this, the external and internal eavesdropping scenarios are both discussed in the following parts.

\subsubsection{External Eavesdropping Scenario}
The baleful external Eve tends to overhear the legitimate channels. In this case, the received signal at external Eve is given by ${y_{EE}} = {\mathbf{h}}_{re}^H{\mathbf{\Theta }}{{\mathbf{h}}_{br}} {\sum\limits_{i = 1}^K {\sqrt {{a_i}{P_s}} {x_i}}} + {n_e}$, where the ${{\mathbf{h}}_{re}} = {\left[ {h_{re}^1 \cdots h_{re}^m \cdots h_{re}^M} \right]^H}$, the $h_{re}^m = \sqrt {d_{re}^{ - \alpha }} \tilde h_{re}^m \sim \mathcal{C}\mathcal{N}\left( {0,{N_{re}}} \right)$ denotes the complex channel coefficient from the \emph{m}-th reflecting element of RIS to the Eve and $\tilde h_{rk}^m \sim \mathcal{C}\mathcal{N}\left( {0,1} \right)$. The ${d_{re}}$ represents the distance from RIS to the Eve. The ${n_e}$ represents the AWGN of Eve with mean power parameter ${N_E}$.
After receiving the superposed signal, the Eve starts decoding the \emph{k}-th LU's message while regarding the rest of information as interference by applying SIC. As a consequence, the SINR for external Eve to decode the \emph{k}-th LU's information (\emph{k} $ < $ \emph{K}) can be given by
\begin{align}\label{SINR EE decode k}
{\gamma _{EE \to k}} = \mathop {}\nolimits^{} \mathop {}\nolimits^{} \frac{{{{\left| {{\mathbf{h}}_{re}^H{\mathbf{\Theta }}{{\mathbf{h}}_{br}}} \right|}^2}{a_k}{\rho _e}}}{{{\rho _e}{{\left| {{\mathbf{h}}_{re}^H{\mathbf{\Theta }}{{\mathbf{h}}_{br}}} \right|}^2}{{\nu _k}} + \varpi {\rho _e}{{\left| {{h_{ipe}}} \right|}^2} + 1}},
\end{align}
where ${\nu _k} = \sum\nolimits_{i = k + 1}^K {{a_i}}$, ${\rho _e} = \frac{{{P_s}}}{{{N_E}}}$ represents the transmit SNR of eavesdropping channel from BS to Eve. The ${h_{ipe}} \sim \mathcal{C}\mathcal{N}\left( {0,{N_{ipe}}} \right)$ represents the corresponding complex channel parameter of residual interference.

The SINR for external Eve to decode the \emph{K}-th LU's information can be given by
\begin{align}\label{SINR EE decode K}
{\gamma _{EE \to K}} = \frac{{{{\left| {{\mathbf{h}}_{re}^H{\mathbf{\Theta }}{{\mathbf{h}}_{br}}} \right|}^2}{a_K}{\rho _e}}}{{ \varpi {\rho _e}{{\left| {{h_{ipe}}} \right|}^2} + 1}}.
\end{align}
\subsubsection{Internal Eavesdropping Scenario}
Under this case, the most distant user is regarded as an internal Eve since it has the worst channel condition\footnote{
Theoretically, the near user can also act as an internal Eve. At this time, the secrecy rate of the system will deteriorate or even drop below the
target secrecy rate, since it can decode far user's information directly. This situation is even more challenging for the security design of RIS-NOMA networks, which will be considered in our future work.
}\cite{2020XinweiUnifiedNOMA,ding2017spectral}. As a result, the received signal at internal Eve is given by ${y_{IE}} = {\mathbf{h}}_{r1}^H{\mathbf{\Theta }}{{\mathbf{h}}_{br}} {\sum\limits_{i = 1}^K {\sqrt {{a_i}{P_s}} {x_i}}} + {n_e}$. In this case, the SINR for internal Eve to detect the \emph{k}-th LU information (1 $<$ \emph{k} $<$ \emph{K}) is given as follows
\begin{align}\label{SINR IE decode k}
{\gamma _{IE \to k}} = \mathop {}\nolimits^{} \mathop {}\nolimits^{} \frac{{{{\left| {{\mathbf{h}}_{r1}^H{\mathbf{\Theta }}{{\mathbf{h}}_{br}}} \right|}^2}{a_k}{\rho _e}}}{{{\rho _e}{{\left| {{\mathbf{h}}_{r1}^H{\mathbf{\Theta }}{{\mathbf{h}}_{br}}} \right|}^2} {{\nu _k}}  + \bar \omega {\rho _e}{{\left| {{h_{ipe}}} \right|}^2} + 1}}.
\end{align}

Similarly, the SINR for internal Eve to decode the \emph{K}-th LU's information can be given by
\begin{align}\label{SINR IE decode K}
{\gamma _{IE \to K}} = \frac{{{{\left| {{\mathbf{h}}_{r1}^H{\mathbf{\Theta }}{{\mathbf{h}}_{br}}} \right|}^2}{a_K}{\rho _e}}}{{\bar \omega {\rho _e}{{\left| {{h_{ipe}}} \right|}^2} + 1}}.
\end{align}
\subsection{RIS-NOMA with On-off Control}
As revealed in the aforementioned analysis, choosing an appropriate phase shift design is pivotal to evaluating the secure performance of RIS-NOMA networks. Hence, a typical type of phase shift design, i.e., on-off control, is considered in this paper, where each diagonal element of RIS is regarded as 1 (on) or 0 (off)\footnote{As a special case of random phase shifting, the on-off control strategy is easy to deploy and extend to MIMO as well as multi-cell scenarios with less demand for CSI. However, the on-off control may not yield the superior diversity order unless the operating state of the reflecting elements is further optimised on top of this, which is beyond the scope of this work.
} \cite{2020ZhiguoSimpleDesignIRSNOMA,2022XinweiIRSNOMA}. More specifically, suppose that \emph{M} = \emph{PQ}, where \emph{P} and \emph{Q} are both positive integers. We define ${\mathbf{V}} = {{\mathbf{I}}_P} \otimes {{\mathbf{1}}_Q} \in {\mathbb{C}^{M \times P}}$ and the \emph{p}-th column of ${\mathbf{V}}$ is denoted by ${{\mathbf{v}}_p}$. As a result, the cascaded Rayleigh channels ${{\mathbf{h}}_{_{rk}}^H{\mathbf{\Theta }}{{\mathbf{h}}_{br}}}$ and ${{\mathbf{h}}_{re}^H{\mathbf{\Theta }}{{\mathbf{h}}_{br}}}$ can be defined as ${H_k} \triangleq {\mathbf{h}}_{rk}^H{\mathbf{\Theta }}{{\mathbf{h}}_{br}} = {\mathbf{v}}_p^H{{\mathbf{D}}_{rk}}{{\mathbf{h}}_{br}}$ and ${H_e} \triangleq {\mathbf{h}}_{re}^H{\mathbf{\Theta }}{{\mathbf{h}}_{br}} = {\mathbf{v}}_p^H{{\mathbf{D}}_{re}}{{\mathbf{h}}_{br}}$ with $k = 1,2, \cdots ,K$, where ${{\mathbf{D}}_{rk}}$ and ${{\mathbf{D}}_{re}}$ are diagonal matrices with their diagonal elements composed from ${\mathbf{h}}_{rk}^H$ and ${\mathbf{h}}_{re}^H$, respectively.

With the design of on-off control, equation (\ref{SINR k decode g}) and (\ref{SINR K decode K}) can be rewritten as
\begin{align}\label{SINR k decode g v2}
{{\tilde \gamma }_{k \to g}} = \mathop {}\nolimits^{} \mathop {}\nolimits^{} \frac{{\rho {{\left| {{{{H}}_{k}}} \right|}^2}{a_g}}}{{\rho {{\left| {{{{H}}_{k}}} \right|}^2} {{{\nu _g}}}+ \varpi \rho {{\left| {{h_{ipu}}} \right|}^2} + 1}},
\end{align}
and
\begin{align}\label{SINR K decode K v2}
{{\tilde \gamma }_K} = \frac{{\rho {{\left| {{{{H}}_{K}}} \right|}^2}{a_K}}}{{\varpi \rho {{\left| {{h_{ipu}}} \right|}^2} + 1}},
\end{align}
respectively.
For the sake of facilitating subsequent analysis, several channel statistical formulas are derived in the following.
\begin{lemma} \label{Lemma1}
Under the condition of on-off control, the CDF of SINR for the k-th LU to decode the g-th LU's information (1 $ \leqslant $ g $ \leqslant $ k) with ipSIC can be given by
\begin{small}
\begin{align}\label{CDF SINR k decode g ipSIC}
\begin{split}
  F_{{{\tilde \gamma }_{k \to g}}}^{ipSIC}\left( x \right) \approx& \kappa \sum\limits_{l = 0}^{K - k} {\sum\limits_{d = 0}^D {
  {K - k} \choose
  l } } \frac{{{{\left( { - 1} \right)}^l}{G_d}}}{{k + l}}\left[ {1 - \frac{2}{{\Gamma \left( Q \right)}}} \right. \\
   &\times {\left. {{{\left( {\frac{{{\zeta _1}x}}{{{a_g} - {\nu _g}x}}} \right)}^{\frac{Q}{2}}}{K_Q}\left( 2\sqrt{{\frac{{{\zeta _1}x}}{{{a_g} - {\nu _g}x}}}} \right)} \right]^{k + l}},
\end{split}
\end{align}
\end{small}where $\kappa  = \frac{{K!}}{{\left( {K - k} \right)!\left( {k - 1} \right)!}}$, ${\zeta _1} = \frac{{\left( {\varpi \rho {N_{ipu}}{\tau _d} + 1} \right)}}{{\rho {N_{br}}{N_{rk}}}}$, ${G_d} = \frac{{{{\left( {D!} \right)}^2}}}{{{\tau _d}{{\left[ {L_{_D}^\prime \left( {{\tau _d}} \right)} \right]}^2}}}$ and ${{\tau _d}}$ represent the weight of Gauss-Laguerre integration and the d-th zero point of Laguerre polynomial ${L_D}\left( {{\tau _d}} \right)$ with d = 1,2,3,...,D, respectively. D denotes a complexity accuracy tradeoff factor and the equal sign in (\ref{CDF SINR k decode g ipSIC}) can be established when D approaches infinity. ${K_Q}\left( x \right)$ represents the modified Bessel function of the second kind \emph{\cite[Eq. (8.432.1)]{gradvstejn2000table}}. $\Gamma \left( Q \right)$ is the gamma function \emph{\cite[Eq. (8.310.1)]{gradvstejn2000table}}.
\end{lemma}
\begin{proof}
See Appendix~A.
\end{proof}

When $\varpi $=0, the CDF of SINR for the \emph{k}-th LU to decode the \emph{g}-th LU's information (1 $ \leqslant $ \emph{g} $ \leqslant $ \emph{k}) with pSIC can be given by
\begin{small}
\begin{align}\label{CDF SINR k decode g pSIC}
\begin{split}
  F_{{{\tilde \gamma }_{k \to g}}}^{pSIC}\left( x \right) \approx& \kappa \sum\limits_{l = 0}^{K - k} {
  {K - k} \choose
  l } \frac{{{{\left( { - 1} \right)}^l}}}{{k + l}}\left[ {1 - \frac{2}{{\Gamma \left( Q \right)}}} \right. \\
  &\times {\left. { {{\left( {\frac{{x{\zeta _2}^{ - 1}}}{{ {{a_g} - {{{\nu _g}}x} } }}} \right)}^{\frac{Q}{2}}}{K_Q}\left( 2\sqrt{{\frac{{x{\zeta _2}^{ - 1}}}{{ {{a_g} - {{\nu _g}x} } }}}} \right)} \right]^{k + l}},
\end{split}
\end{align}
\end{small}where ${\zeta _2} = \rho {N_{br}}{N_{rk}}$.

\begin{lemma} \label{Lemma2}
After striking out the previous (K-1) LUs' messages with ipSIC, the CDF of SINR for the K-th LU to decode its own information can be given by
\begin{small}
\begin{align}\label{CDF SINR K decode K ipSIC}
F_{{{\tilde \gamma }_K}}^{ipSIC}\left( x \right) \approx& \kappa \sum\limits_{l = 0}^{K - k} {\sum\limits_{d = 0}^D {{
  {K - k} \choose
  l } \frac{{{{\left( { - 1} \right)}^l}}}{{k + l}}{G_d}} }\notag\\  &\times {\left\{ {1 - \frac{2}{{\Gamma \left( Q \right)}}{{\left[ {h\left( x \right)} \right]}^{\frac{Q}{2}}}{K_Q}\left[ {2\sqrt {h\left( x \right)} } \right]} \right\}^{k + l}},
\end{align}
\end{small}
where $h\left( x \right) = \frac{{\left( {\varpi \rho {N_{ipu}}{\tau _d} + 1} \right)x}}{{{\zeta _2}{a_K}}}$.
\end{lemma}

When $\varpi $=0, the CDF of SINR for the \emph{K}-th LU to decode it own information with pSIC can be given by
\begin{small}
\begin{align}\label{CDF SINR K decode K pSIC}
  F_{{{\tilde \gamma }_K}}^{pSIC}\left( x \right) \approx& \kappa \sum\limits_{l = 0}^{K - k} {{
  {K - k} \choose
  l }\frac{{{{\left( { - 1} \right)}^l}}}{{k + l}}}\notag \\ &\times {\left[ {1 - \frac{2}{{\Gamma \left( Q \right)}}{{\left( {\frac{x}{{{\zeta _2}{a_K}}} } \right)}^{\frac{Q}{2}}}{K_Q}\left( {2\sqrt {\frac{x}{{{\zeta _2}{a_K}}} }} \right)} \right]^{k + l}}.
\end{align}
\end{small}

\section{Secrecy Performance Evaluation}\label{Section_III}
In this section, the SOP is selected as a crucial metric to evaluate the security performance of RIS-NOMA networks with the existence of external and internal Eves. In order to trace the asymptotic security features, secrecy diversity order is obtained in the high SNR region according to the analytical outcomes.
\subsection{Statistical Models for Wiretap Channels}
To further explore the influence of Eve on system security, several essential statistical formulas for eavesdropping channels are derived as follows.
\subsubsection{External Eavesdropping Scenario}
In eavesdropping scenarios, the cascaded wiretap channels can also be recast according to on-off control. Hence, equations (\ref{SINR EE decode k}) and (\ref{SINR EE decode K}) can be rewritten as
\begin{align}\label{SINR EE decode k v2}
{{\tilde \gamma }_{EE \to k}} = \frac{{{{\left| {{{{H}}_{e}}} \right|}^2}{a_k}{\rho _e}}}{{{\rho _e}{{\left| {{{{H}}_{e}}} \right|}^2} {\nu _k} + \varpi {\rho _e}{{\left| {{h_{ipe}}} \right|}^2} + 1}},
\end{align}
and
\begin{align}\label{SINR EE decode K v2}
{{\tilde \gamma }_{EE \to K}} = \frac{{{{\left| {{{{H}}_{e}}} \right|}^2}{a_K}{\rho _e}}}{{\varpi {\rho _e}{{\left| {{h_{ipe}}} \right|}^2} + 1}},
\end{align}
respectively.
\begin{lemma} \label{Lemma3}
Under the condition of on-off control, the PDF of SINR for the external Eve to decode the k-th LU's information with ipSIC can be given by
\begin{small}
\begin{align}\label{PDF SINR EE decode k ipSIC v2}
\begin{gathered}
  f_{{{\tilde \gamma }_{EE \to k}}}^{ipSIC}\left( x \right) \approx \sum\limits_{d = 0}^D {{G_d}\frac{{{a_k}{{\left[ {g\left( x \right)} \right]}^{\frac{Q}{2}}}}}{{x\left( {{a_k} - {{{\nu _k}x} }} \right)\Gamma \left( Q \right)}}}  \hfill \\
  {\text{ }} \times \left\langle {{{\left[ {g\left( x \right)} \right]}^{\frac{1}{2}}}\left\{ {{K_{Q - 1}}\left[ {2\sqrt {g\left( x \right)} } \right] + {K_{Q + 1}}\left[ {2\sqrt {g\left( x \right)} } \right]} \right\} - Q} \right\rangle,  \hfill \\
\end{gathered}
\end{align}
\end{small}
where $g\left( x \right) = \frac{{\left( {\varpi {\rho _e}{N_{ipe}}{\tau _d} + 1} \right)x}}{{\left( {{a_k} - {{\nu _k}x} } \right){\rho _e}{N_{br}}{N_{re}}}}$.
\begin{proof}
According to the on-off control principle, the CDF of ${{\tilde \gamma }_{EE \to k}}$  can be shown as
\begin{footnotesize}
\begin{align}\label{CDF SINR EE decode k v2-temp}
F_{{{\tilde \gamma }_{EE \to k}}}^{ipSIC}\left( x \right) = \int_0^\infty  {{f_{{{\left| {{h_{ipe}}} \right|}^2}}}\left( y \right)} {F_{{{\left| {{{{H}}_{e}}} \right|}^2}}}\left[ {\frac{{{\rho _e}^{ - 1}\left( {\varpi {\rho _e}y + 1} \right)x}}{{ {{a_k}} - {{{\nu _k}x} } }}} \right]dy.
\end{align}
\end{footnotesize}
Then equation (\ref{CDF SINR EE decode k v2-temp}) can be handled referring to the calculation of (\ref{PDF of cascaded channel NotORD}) - (\ref{CDF of cascaded channel ORD}) in Appendix A as follows
\begin{small}
\begin{align}\label{CDF SINR EE decode k v2}
F_{{{\tilde \gamma }_{EE \to k}}}^{ipSIC}\left( x \right) \approx \sum\limits_{d = 0}^D {{G_d}\left\{ {1 - \frac{{ 2}}{{\Gamma \left( Q \right)}}{{\left[ {g\left( x \right)} \right]}^{\frac{Q}{2}}} {K_Q}\left[ {2\sqrt {g\left( x \right)} } \right]} \right\}} .
\end{align}
\end{small}
By employing derivative operation to $F_{{{\tilde \gamma }_{EE \to k}}}^{ipSIC}\left( x \right)$, we can acquire equation (\ref{PDF SINR EE decode k ipSIC v2}) and the proof is completed.
\end{proof}
\end{lemma}
When $\varpi $=0, the CDF of SINR for the external Eve to decode the \emph{k}-th LU's information with pSIC can be given by
\begin{align}\label{PDF SINR EE decode k pSIC v2}
F_{{{\tilde \gamma }_{EE \to k}}}^{pSIC}\left( x \right) =& 1 - \frac{2}{{\Gamma \left( Q \right)}}{\left[ {\frac{x}{{ \left({{a_k}} - {{{\nu _k}x} } \right) {\rho _e}}}} \right]^{\frac{Q}{2}}}\notag \\ &\times {K_Q}\left[ {2\sqrt {\frac{x}{{ \left({{a_k}} - {{{\nu _k}x} } \right) {\rho _e}}}} } \right].
\end{align}
\subsubsection{Internal Eavesdropping Scenario}
Considering on-off control, the cascaded internal eavesdropping channels can be transformed into the following
\begin{align}\label{SINR IE decode k v2}
{{\tilde \gamma }_{IE \to k}} = \frac{{{{\left| {{{H}_{1}}} \right|}^2}{a_k}{\rho _e}}}{{{\rho _e}{{\left| {{{H}_{1}}} \right|}^2} {{{\nu _k}x} }  + \bar \omega {\rho _e}{{\left| {{h_{ipe}}} \right|}^2} + 1}},
\end{align}
and
\begin{align}\label{SINR IE decode K v2}
{{\tilde \gamma }_{IE \to K}} = \frac{{{{\left| {{{H}_{1}}} \right|}^2}{a_K}{\rho _e}}}{{\bar \omega {\rho _e}{{\left| {{h_{ipe}}} \right|}^2} + 1}},
\end{align}
respectively.
\begin{lemma} \label{Lemma4}
Under the condition of on-off control, the PDF of SINR for the first user to decode the k-th LU's information with ipSIC is given by (\ref{PDF IE decode k ipSIC v2}), shown at the top of next page, where $q\left( x \right) = \frac{{\left( {\varpi {\rho _e}{N_{ipe}}{\tau _d} + 1} \right)x}}{{\left( {{a_k} - {{\nu _k}x} } \right){\rho _e}{N_{br}}{N_{r1}}}}$.
\end{lemma}
\begin{figure*}[!t]
%\begin{small}
\begin{small}
\begin{align}\label{PDF IE decode k ipSIC v2}
f_{{{\tilde \gamma }_{IE \to k}}}^{ipSIC}\left( x \right) \approx& \kappa \sum\limits_{l = 0}^{K - k} {\sum\limits_{d = 0}^D {{G_d}{{
  {K - k} \choose
  l }}{{\left( { - 1} \right)}^l}{{\{ {1 - \frac{{2}}{{\Gamma \left( Q \right)}}{{\left[ {q\left( x \right)} \right]}^{\frac{Q}{2}}}{K_Q}\left[ {2\sqrt {q\left( x \right)} } \right]} \}}^{k + l - 1}}} }  \notag \\
 &\times \frac{{{a_k}{{\left[ {q\left( x \right)} \right]}^{\frac{Q}{2}}}}}{{\Gamma \left( Q \right)\left( {{a_k} - {{{\nu _k}x} }} \right)x}}\left\langle {\sqrt {q\left( x \right)} \left\{ {{K_{Q - 1}}\left[ {2\sqrt {q\left( x \right)} } \right] + {K_{Q + 1}}\left[ {2\sqrt {q\left( x \right)} } \right]} \right\} - Q{K_Q}\left[ {2\sqrt {q\left( x \right)} } \right]} \right\rangle .
\end{align}
\end{small}
\begin{small}
\begin{align}\label{PDF IE decode k pSIC v2}
f_{{{\tilde \gamma }_{IE \to k}}}^{pSIC}\left( x \right) \approx& \kappa \sum\limits_{l = 0}^{K - k} {{
  {K - k} \choose
  l }} {\left( { - 1} \right)^l}{\{ 1 - \frac{2}{{\Gamma \left( Q \right)}}{\left[ {p\left( x \right)} \right]^{\frac{Q}{2}}}{K_Q}\left[ {2\sqrt {p\left( x \right)} } \right]\} ^{k + l - 1}} \notag \\
 &\times \frac{{{a_k}{{\left[ {p\left( x \right)} \right]}^{Q - 1}}}}{{\Gamma \left( Q \right)\left( {{a_k} - {{\nu _k}x} } \right)x}}\left\langle {\sqrt {p\left( x \right)} \{ {K_{Q - 1}}\left[ {2\sqrt {p\left( x \right)} } \right] + {K_{Q + 1}}\left[ {2\sqrt {p\left( x \right)} } \right]\}  - Q{K_Q}\left[ {2\sqrt {p\left( x \right)} } \right]} \right\rangle .
\end{align}
\end{small}
\hrulefill \vspace*{0pt}
\end{figure*}

When $\varpi $=0, the PDF of SINR for the internal Eve to decode the \emph{k}-th LU's information with pSIC is given by (\ref{PDF IE decode k pSIC v2}), shown at the top of next page, where $p\left( x \right) = \frac{x}{{\left( {{a_k} - {{\nu _k}x} } \right){\rho _e}{N_{br}}{N_{r1}}}}$.

\subsection{Secrecy Outage Probability}
In this subsection, the approximate and asymptotic SOP expressions for LUs are obtained.
\subsubsection{External Eavesdropping Scenario}
According to \cite{2022TianweiRISNOMAPLS}, the secrecy rate for the \emph{k}-th LU is defined as
\begin{align}\label{secrecy rate for k EE}
{C_k^{EE}} = \left[ {{{\log }_2}\left( {1 + {{\tilde \gamma }_k}} \right)} \right.{\left. { - {{\log }_2}\left( {1 + {{\tilde \gamma }_{EE \to k}}} \right)} \right]^ + },
\end{align}
where ${{{\tilde \gamma }_k}}$ denotes the SINR for \emph{k}-th LU to decode its own information with on-off control. ${R_k^{EE}}$ is the target secrecy rate of the \emph{k}-th LU. As a consequence, the secrecy outage event occurs when ${C_k^{EE}} < {R_k^{EE}}$ and the relevant SOP expression can be given by
\begin{align}\label{SOP k temp}
P_{out}^{k,EE}\left( {{R_k^{EE}}} \right) &= P\left( {{C_k^{EE}} < {R_k^{EE}}} \right) \notag \\
&= P\left[ {{{\tilde \gamma }_k} < {2^{{R_k}}}\left( {1 + {{\tilde \gamma }_{EE \to k}}} \right) - 1} \right] .
\end{align}
Note that (\ref{SOP k temp}) shows the SOP expression for the \emph{k}-th LU is determined by both its own SINR and Eve's wiretapping SINR in RIS-NOMA secure networks, which is quite different from the outage probability expression based on the SINR of a single user in the traditional RIS-NOMA networks.
According to Lemma \ref{Lemma1} and Lemma \ref{Lemma3}, the SOP for the \emph{k}-th LU can be expressed as
\begin{align}\label{SOP k}
P_{_{out,k}}^{\phi,EE} \left( {{R_k^{EE}}} \right) = \int_0^\infty  {f_{_{{{\tilde \gamma }_{EE \to k}}}}^\phi \left( x \right)} F_{_{{{\tilde \gamma }_k}}}^\phi \left[ {{2^{{R_k^{EE}}}}\left( {1 + x} \right) - 1} \right]dx,
\end{align}
where $\phi  \in \left\{ {ipSIC,pSIC} \right\}$. The case of ipSIC and pSIC are derived with the assistance of (\ref{CDF SINR k decode g ipSIC}), (\ref{PDF SINR EE decode k ipSIC v2}) and (\ref{CDF SINR k decode g pSIC}), (\ref{PDF SINR EE decode k pSIC v2}), respectively.
However, dealing with the complex integrations could be pretty tough. To solve this problem, a reasonable approximation method is introduced to acquire the approximate expressions of SOP with verified accuracy, which can provide guidance for the theoretical understanding and practical applications.

\begin{theorem}\label{Theorem1}
Under the condition of on-off control, the approximate expression of SOP for the external Eve to decode information of the k-th LU with ipSIC can be given by (27), shown at the top of next page, where ${\eta _s} = {2^{{R_k^{EE}}}}\left( {1 + \frac{{Q{N_{br}}{N_{re}}{a_k}{\rho _e}}}{{ \varpi {\rho _e}{N_{ipe}}{\tau _s} + 1}}} \right) - 1$, ${G_s} = \frac{{{{\left( {S!} \right)}^2}}}{{{\tau _s}{{\left[ {L_{_S}^\prime \left( {{\tau _s}} \right)} \right]}^2}}}$ and ${{\tau _s}}$ represent the weight of Gauss-Laguerre integration and the s-th zero point of Laguerre polynomial ${L_S}\left( {{\tau _s}} \right)$ with s = 1,2,3,...,S, respectively. S denotes a complexity accuracy tradeoff factor and the equal sign in (\ref{SOP EE decode k ipSIC}) can be established when S and D both approach infinity.
\end{theorem}
\begin{figure*}[!t]
\begin{small}
\begin{align}\label{SOP EE decode k ipSIC}
\begin{gathered}
  P_{out,k}^{ipSIC,EE}\left( {{R_k^{EE}}} \right) \approx \sum\limits_{s = 0}^S {\sum\limits_{l = 0}^{K - k} {\sum\limits_{d = 0}^D {\frac{\kappa {G_s}{G_d}{{{\left( { - 1} \right)}^l}}}{{\left( {k + l} \right)}}} } } {
  {K - k} \choose
  l } \hfill
 {\left\{ {1 - \frac{2}{{\Gamma \left( Q \right)}}{{\left[ {\frac{{{\eta _s}{\zeta _1}}}{{\left( {{a_k} - {\eta _s}{\nu _k}} \right)}}} \right]}^{\frac{Q}{2}}}{K_Q}\left[ {2\sqrt {\frac{{{\eta _s}{\zeta _1}}}{{\left( {{a_k} - {\eta _s}{{{\nu _k}}} } \right)}}} } \right]} \right\} ^{k + l}} \hfill \\
\end{gathered} .
\end{align}
\begin{align}\label{SOP IE decode k ipSIC}
\begin{gathered}
  P_{out,k}^{ipSIC,IE}\left( {{R_k^{IE}}} \right) \approx \sum\limits_{s = 0}^S {\sum\limits_{l = 0}^{K - k} {\sum\limits_{d = 0}^D {\frac{\kappa {G_s}{G_d}{{{\left( { - 1} \right)}^l}}}{{\left( {k + l} \right)}}} } } {
  {K - k} \choose
  l } \hfill
   {\left\{ {1 - \frac{2}{{\Gamma \left( Q \right)}}{{\left[{\frac{{{\zeta _1}{\vartheta _s}}}{{\left( {{a_k} - {\vartheta _s}{\nu _k} } \right)}}} \right]}^{\frac{Q}{2}}}{K_Q}\left[ {2\sqrt {\frac{{{\zeta _1}{\vartheta _s}}}{{\left( {{a_k} - {\vartheta _s}{{{\nu _k}}} } \right)}}} } \right]} \right\} ^{k + l}} \hfill \\
\end{gathered} .
\end{align}
\end{small}
\hrulefill \vspace*{0pt}
\end{figure*}
\begin{proof}
See Appendix~B.
\end{proof}

When $\varpi $=0, the approximate expression of SOP for the external Eve to decode the \emph{k}-th LU's signal with pSIC is shown as
\begin{small}
\begin{align}\label{SOP EE decode k pSIC}
  P_{out,k}^{pSIC,EE}\left( {R_k^{EE}} \right) =& \kappa \sum\limits_{l = 0}^{K - k} {{
  {K - k} \choose
  l }} \frac{{{{\left( { - 1} \right)}^l}}}{{k + l}} \notag \\
&\times {\left\{ {1 - \frac{2}{{\Gamma \left( Q \right)}}{{\left[ {s\left( \psi  \right)} \right]}^{\frac{Q}{2}}}{K_Q}\left[ {2\sqrt {s\left( \psi  \right)} } \right]} \right\}^{k + l}},
\end{align}
\end{small}where $s\left( \psi  \right) = \frac{{\psi \zeta _2^{ - 1}}}{{{a_k} - \psi {{\nu _k}} }}$, $\psi  = {2^{{R_k^{EE}}}}\left( {1 + \frac{{Q{N_{br}}{N_{re}}{a_k}{\rho _e}}}{{{\rho _e}Q{N_{br}}{N_{re}}{{\nu _k}}  + 1}}} \right) - 1$.

\subsubsection{Internal Eavesdropping Scenario}
We suppose that the most distant user is regarded as the internal Eve and tends to wiretap other LU's information. The secrecy rate for the \emph{k}-th LU ($k = 2,3, \cdots,K$) can be expressed as
\begin{align}\label{secrecy rate for k IE}
C_{_k}^{IE} = \left[ {{{\log }_2}\left( {1 + {{\tilde \gamma }_k}} \right)} \right.{\left. { - {{\log }_2}\left( {1 + {{\tilde \gamma }_{IE \to k}}} \right)} \right]^ + }.
\end{align}
Hence, the secrecy outage event occurs when the secrecy rate $C_{_k}^{IE}$ is less than the target secrecy rate $R_{_k}^{IE}$ and the SOP expression for the \emph{k}-th LU is given as follows
\begin{align}\label{SOP IE decode k ipSIC temp}
  P_{_{out,k}}^{\phi ,IE}\left( {{R_k}} \right) &= P\left[ {{{\tilde \gamma }_k} < {2^{{R_k}}}\left( {1 + {{\tilde \gamma }_{IE \to k}}} \right) - 1} \right] \notag \\
&= \int_0^\infty  {f_{_{{{\tilde \gamma }_{IE \to k}}}}^\phi \left( x \right)} F_{_{{{\tilde \gamma }_k}}}^\phi \left[ {{2^{R_{_k}^{IE}}}\left( {1 + x} \right) - 1} \right]dx.
\end{align}
\begin{theorem}\label{Theorem2}
Under the condition of on-off control, the approximate expression of SOP for the internal Eve to decode information of the k-th LU with ipSIC can be given by (\ref{SOP IE decode k ipSIC}), shown at the top of next page, where ${\vartheta _s} = {2^{{R_{_k}^{IE}}}}\left( {1 + \frac{{Q{N_{br}}{N_{r1}}{a_k}{\rho _e}}}{{{\rho _e}Q{N_{br}}{N_{r1}}{{\nu _k}}  + \varpi {\rho _e}{N_{ipe}}{\tau _s} + 1}}} \right) - 1$.
\end{theorem}

When $\varpi $=0, the approximate expression of SOP for the internal Eve to decode signal of the \emph{k}-th LU with pSIC is shown as
\begin{small}
\begin{align}\label{SOP IE decode k pSIC}
P_{out,k}^{pSIC,IE} =& \kappa \sum\limits_{l = 0}^{K - k} {
{K - k} \choose
l } \frac{{{{\left( { - 1} \right)}^l}}}{{k + l}} \notag \\
&\times {\left[ {1 - \frac{2}{{\Gamma \left( Q \right)}}{{\left\{ {u\left( \varsigma  \right)} \right\}}^{\frac{Q}{2}}}{K_Q}\left( {2\sqrt {u\left( \varsigma  \right)} } \right)} \right]^{k + l}},
\end{align}
\end{small}where $u\left( \varsigma  \right) = \frac{{\varsigma \zeta _2^{ - 1}}}{{{a_k} - \varsigma {{\nu _k}} }}$ and $\varsigma  = {2^{R_k^{IE}}}\left( {1 + \frac{{Q{N_{br}}{N_{r1}}{a_k}{\rho _e}}}{{{\rho _e}Q{N_{br}}{N_{r1}}{{\nu _k}}  + 1}}} \right) - 1$.

\subsection{Secrecy Diversity Order}
To gain more insights, the asymptotic performance in the high SNR region of SOP is investigated. The secrecy diversity order is expressed as follows \cite{2022XinweiIRSNOMA}
\begin{align}\label{div}
div =  - \mathop {\lim }\limits_{\rho  \to \infty } \frac{{\log \left[ {P_{out}^{asy} \left( \rho  \right)} \right]}}{{\log \rho }},
\end{align}
where ${P_{out}^\infty \left( \rho  \right)}$ denotes the asymptotic SOP with factor $\rho $. The asymptotic behaviors for LUs in both external and internal wiretap scenarios are analysed in the following.
\begin{corollary}\label{Corollary1}
Under the condition of external eavesdropping, the asymptotic SOP at the high SNR regime with ipSIC of the k-th LU is shown by (\ref{AsySOP EE decode k ipSIC}) at the top of the next page
\begin{figure*}[!t]
\begin{small}
\begin{align}\label{AsySOP EE decode k ipSIC}
  P_{asy,k}^{ipSIC,EE}\left( {R_k^{EE}} \right) \approx& \sum\limits_{s = 0}^S {\sum\limits_{l = 0}^{K - k} {\sum\limits_{d = 0}^D  {{\kappa {G_s}{G_d}}}{}{{
  {K - k} \choose
  l }}\frac{{{{\left( { - 1} \right)}^l}}}{{k + l}}} }  \notag \\
&\times {\left\{ {1 - \frac{2}{{\Gamma \left( Q \right)}}{{\left[ {\frac{{{\eta _s}{N_{ipu}}\varpi {\tau _d}}}{{\left( {{a_k} - {\eta _s}{{\nu _k}} } \right){N_{br}}{N_{rk}}}}} \right]}^{\frac{Q}{2}}}{K_Q}\left[ {2\sqrt {\frac{{{\eta _s}{N_{ipu}}\varpi {\tau _d}}}{{\left( {{a_k} - {\eta _s}{{\nu _k}} } \right){N_{br}}{N_{rk}}}}} } \right]} \right\} ^{k + l}}.
\end{align}
\end{small}
\begin{small}
\begin{align}\label{AsySOP IE decode k ipSIC}
P_{asy,k}^{ipSIC,IE}\left( {R_k^{IE}} \right) \approx& \sum\limits_{s = 0}^S {\sum\limits_{l = 0}^{K - k} {\sum\limits_{d = 0}^D {{{\kappa {G_s}{G_d}}}{}} } } {{
  {K - k} \choose
  l }}\frac{{{{\left( { - 1} \right)}^l}}}{{k + l}}\notag \\ &\times {\left\{ {1 - \frac{2}{{\Gamma \left( Q \right)}}{{\left[ {\frac{{{\vartheta _s}\varpi {N_{ipu}}{\tau _d}}}{{\left( {{a_k} - {\vartheta _s}{{\nu _k}} } \right){N_{br}}{N_{rk}}}}} \right]}^{\frac{Q}{2}}}{K_Q}\left[ {2\sqrt {\frac{{{\vartheta _s}\varpi {N_{ipu}}{\tau _d}}}{{\left( {{a_k} - {\vartheta _s}{{\nu _k}} } \right){N_{br}}{N_{rk}}}}} } \right]} \right\}^{k + l}}.
\end{align}
\end{small}
\hrulefill \vspace*{0pt}
\end{figure*}
\begin{proof}
According to (\ref{SINR k decode g v2}), the CDF of ${{\tilde \gamma }_k}$ can be written as follows
\begin{small}
\begin{align}\label{AsyCDF of SINR k temp }
F_{{{\tilde \gamma }_k}}^{ipSIC}\left( x \right) &= P\left( {{{\tilde \gamma }_k} < x} \right) \notag \\
 &= P\left( {\frac{{\rho {{\left| {{H_k}} \right|}^2}{a_k}}}{{\rho {{\left| {{H_k}} \right|}^2}{\nu _k} + \varpi \rho {{\left| {{h_{ipu}}} \right|}^2} + 1}} < x} \right) \notag \\
 &= \int_0^\infty  {{f_{{{\left| {{h_{ipu}}} \right|}^2}}}\left( y \right){F_{{{\left| {{H_k}} \right|}^2}}}\left[ {\frac{{x\left( {\varpi \rho y + 1} \right)}}{{\rho \left( {{a_k} - {\nu _k}x} \right)}}} \right]} dy.
\end{align}
\end{small}
Let $\Upsilon  = \frac{{x\left( {\varpi \rho y + 1} \right)}}{{\rho \left( {{a_k} - {\nu _k}x} \right)}}$ and $\Upsilon  \approx \frac{{\varpi xy}}{{{a_k} - {\nu _k}x}}$ when $\rho  \to \infty $. Upon substituting $\Upsilon $ into (\ref{CDF of k decode g ipSIC temp}), we have
\begin{small}
\begin{align}\label{AsyCDF of SINR for k}
\begin{gathered}
  F_{{asy},{{\tilde \gamma }_k}}^{ipSIC} \left( x \right) = \frac{1}{{{N_{ipu}}}}\int_0^\infty  {{e^{ - \frac{y}{{{N_{ipu}}}}}}}  \hfill
   F_{_{{asy},{{\left| {{{{H}}_{k}}} \right|}^2}}}^{ipSIC} \left(\Upsilon \right)dy.
\end{gathered}
\end{align}
\end{small}
With the help of Gauss-Laguerre integration as well as (\ref{CDF of cascaded channel ORD_temp}), we can acquire (\ref{AsySOP EE decode k ipSIC}) and the proof is completed.
\end{proof}
\end{corollary}
\begin{remark}\label{Remark1}
Upon substituting (\ref{AsySOP EE decode k ipSIC}) into (\ref{div}), we find that the secrecy diversity order in external eavesdropping scenario with ipSIC of the k-th LU equals zero, which is due to the negative effects caused by ipSIC.
\end{remark}

When $\varpi $=0, the asymptotic SOPs of the \emph{k}-th LU at the high SNR regime for \emph{Q} = 1 and \emph{Q} $ \geqslant $ 2 with pSIC are expressed as
\begin{small}
\begin{align}\label{AsySOP EE decode k pSIC M1}
P_{asy,k}^{pSIC,EE}\left( {R_k^{EE}} \right) = \frac{\kappa }{k}{\left\langle {1 - \frac{1}{{\Gamma \left( Q \right)}}\left\{ {1 + s\left( \psi  \right)\ln \left[ {s\left( \psi  \right)} \right]} \right\}} \right\rangle ^k},
\end{align}
\end{small}
and
\begin{small}
\begin{align}\label{AsySOP EE decode k pSIC M2}
P_{asy,k}^{pSIC,EE}\left( {R_k^{EE}} \right) = \frac{\kappa }{k}{\left[ {\frac{{s\left( \psi  \right)}}{{Q - 1}}} \right]^k},
\end{align}
\end{small}
respectively.
\emph{\begin{proof}
In the case of Q = 1 and Q $ \geqslant $ 2, the modified Bessel function ${K_Q}\left( x \right)$ can be approximated as ${K_1}\left( x \right) \approx \frac{1}{x} + \frac{x}{2}\ln \left( {\frac{x}{2}} \right)$ and ${K_Q}\left( x \right) \approx \frac{1}{2}\left[ {\frac{{\left( {Q - 1} \right)!}}{{{{\left( {x/2} \right)}^Q}}} - \frac{{\left( {Q - 2} \right)!}}{{{{\left( {x/2} \right)}^{Q - 2}}}}} \right]$, respectively \emph{\cite{2020ZhiguoSimpleDesignIRSNOMA}}.
Upon substituting the approximation formulas into (\ref{CDF SINR K decode K ipSIC}) and keep the first term (l = 0) of the summation formula, the CDFs of SINR for the k-th LU to decode its own message are transformed into
\begin{small}
\begin{align}\label{AsyCDF k decode k pSIC M1}
F_{_{asy,{{\hat \gamma }_k}}}^{pSIC}\left( x \right) =& \kappa \sum\limits_{l = 0}^{K - k} {{
  {K - k} \choose
  l }} \frac{{{{\left( { - 1} \right)}^l}}}{{k + l}}\notag \\ &\times {\left\langle {1 - \frac{{{1}}}{{\Gamma \left( Q \right)}}\left\{ {1 +{\left[ {j\left( x \right)} \right]} \ln \left[ {j\left( x \right)} \right]} \right\}} \right\rangle ^{k}},
\end{align}
\end{small}
and
\begin{small}
\begin{align}\label{AsyCDF k decode k pSIC M2}
F_{_{asy,{{\hat \gamma }_k}}}^{pSIC}\left( x \right) =& \kappa \sum\limits_{l = 0}^{K - k} {{
  {K - k} \choose
  l }} \frac{{{{\left( { - 1} \right)}^l}}}{{k + l}}\notag \\ &\times{\left[ {\frac{x}{{\left( {{a_k} - {{\nu _k}} x} \right){\zeta _2}\left( {Q - 1} \right)}}} \right]^{k}},
\end{align}
\end{small}respectively, where $j\left( x \right) = \frac{x{\zeta _2^{ - 1}}}{{ {{a_k} - {{\nu _k}} x} }}$. With the assistance of (\ref{SOP of EE decode k AppendixB}) and some straightforward calculations, we can acquire (\ref{AsySOP EE decode k pSIC M1}) as well as (\ref{AsySOP EE decode k pSIC M2}) and the proof is completed.
\end{proof}}

\begin{remark}\label{Remark2}
Upon substituting (\ref{AsySOP EE decode k pSIC M1}) and (\ref{AsySOP EE decode k pSIC M2}) into (\ref{div}), we find that the secrecy diversity order in external eavesdropping scenario with pSIC of the k-th LU equals k, which means that the secrecy diversity order of the k-th LU is correlated with channel ordering.
\end{remark}

\begin{corollary}\label{Corollary2}
Under the condition of internal eavesdropping, the asymptotic SOP of the k-th LU at the high SNR regime with ipSIC is shown by (\ref{AsySOP IE decode k ipSIC}) at the top of next page.
\end{corollary}

When $\varpi $=0, the asymptotic SOPs of the \emph{k}-th LU at the high SNR regime with ipSIC for \emph{Q} = 1 and \emph{Q} $ \geqslant $ 2 are expressed as
\begin{small}
\begin{align}\label{AsySOP IE decode k pSIC M1}
P_{asy,k}^{pSIC,IE}\left( {R_k^{IE}} \right) = \frac{\kappa }{k}{\left\langle {1 - \frac{1}{{\Gamma \left( Q \right)}}\left\{ {1 + u\left( \varsigma  \right)\ln \left[ {u\left( \varsigma  \right)} \right]} \right\}} \right\rangle ^k},
\end{align}
\end{small}
and
\begin{small}
\begin{align}\label{AsySOP IE decode k pSIC M2}
P_{asy,k}^{pSIC,IE}\left( {R_k^{IE}} \right) = \frac{\kappa }{k}{\left[ {\frac{\varsigma }{{\left( {Q - 1} \right)\left( {{a_k} - {{\nu _k}\varsigma } } \right){\zeta _2}}}} \right]^k},
\end{align}
\end{small}
respectively.
\begin{remark}\label{Remark3}
Upon substituting (\ref{AsySOP IE decode k ipSIC}) into (\ref{div}), we find that the secrecy diversity order of the k-th LU in internal eavesdropping scenario with ipSIC equals zero, which is due to the negative effects caused by ipSIC. Similarly, plugging (\ref{AsySOP IE decode k pSIC M1}) and (\ref{AsySOP IE decode k pSIC M2}) into (\ref{div}), we see that the secrecy diversity order with pSIC becomes k which is associated with channel ordering.
\end{remark}

\begin{table}[]
\scriptsize
\renewcommand\arraystretch{1.5}
\centering
\caption{The table of Monte Carlo simulation parameters}
\begin{tabular}{|l|l|}
\specialrule{0em}{1pt}{1pt}
\hline
Average SNR of Eve                         & ${\rho _e} = 10$ dB                                                        \\ \hline
%The power allocation coefficient for users  & \begin{tabular}[c]{@{}l@{}}${a_1} = 0.6$\\ ${a_2} = 0.3$\\ ${a_3} = 0.1$\end{tabular}  \\ \hline
\begin{tabular}[c]{@{}l@{}}The power allocation \\ coefficient for users\end{tabular}  & \begin{tabular}[c]{@{}l@{}}$\{a_1,a_2,a_3\} = \{0.6,0.3,0.1\}$\end{tabular}  \\ \hline
\begin{tabular}[c]{@{}l@{}}The power allocation\\ coefficient for users with AN\end{tabular}  & \begin{tabular}[c]{@{}l@{}}$\{a_1,a_2,a_3,a_{an}\} = \{0.4,0.2,0.1,0.3\}$\end{tabular}  \\ \hline
The target secrecy rates for users         & \begin{tabular}[c]{@{}l@{}}$R_k^\varphi  = 0.04$ BPCU\\ $k \in \left\{ {1,2,3} \right\}$, $\varphi  \in \left\{ {EE,IE} \right\}$\end{tabular}        \\ \hline
The distance from BS to RIS                & ${d_{br}} = 3{\text{ m}}$                                                            \\ \hline
The distance from RIS to users             & \begin{tabular}[c]{@{}l@{}}$\{d_{r1},d_{r2},d_{r3}\} = \{6\text{ m},4\text{ m},2\text{ m}\}$\end{tabular} \\ \hline
The distance from RIS to Eve               & ${d_{re}} = 8\text{ m}$                                                        \\ \hline
\specialrule{0em}{1pt}{1pt}
\end{tabular}
\end{table}
\subsection{System Secrecy Outage Probability}
In order to characterize the holistic secrecy performance of RIS-NOMA networks, the system SOP with ipSIC/pSIC can be defined as follows
\begin{align}\label{system SOP}
P_{sys}^{\phi ,\varphi } = 1 - \prod\limits_{k = 1}^K {\left( {1 - P_{out,k}^{\phi ,\varphi }} \right)},
\end{align}where $\phi  \in \left\{ {ipSIC,pSIC} \right\}$, $\varphi  \in \left\{ {EE,IE} \right\}$, ${P_{out,k}^{\phi ,\varphi }}$ can be obtained from (\ref{SOP EE decode k ipSIC}), (\ref{SOP IE decode k ipSIC}), (\ref{SOP EE decode k pSIC}) and (\ref{SOP IE decode k pSIC}).
\subsection{Delay-Limited Transmission}
Given the delay-limited transmission mode, the superposed messages are transmitted at a constant rate, which is limited to the SOP on account of the malicious attack from Eve. Therefore, the secrecy system throughput of RIS-NOMA networks with ipSIC/pSIC under delay-limited transmission mode can be defined based on \cite{2016CaijunNOMA,2013AANasir}
\begin{align}\label{SST define}
R_{T}^{\phi ,\varphi } = \sum\limits_{k = 1}^K {\left( {1 - P_{out,k}^{\phi ,\varphi }} \right)} R_k^\varphi ,
\end{align}
where $\phi  \in \{ ipSIC,pSIC\} $ and $\varphi  \in \{ EE,IE\} $. ${P_{out,k}^{\phi,\varphi} }$ can be acquired from (\ref{SOP EE decode k ipSIC}), (\ref{SOP EE decode k pSIC}), (\ref{SOP IE decode k ipSIC}) and (\ref{SOP IE decode k pSIC}), respectively.
\section{Numerical Results}\label{SectionIV}
In this section, numerical results are provided to substantiate the accuracy of theoretical expressions derived in the aforementioned sections for RIS-NOMA networks. For sake of notational simplicity, Monte Carlo simulation parameters involved are summarized in Table I, where BPCU is an abbreviation for bit per channel use, and the number of Monte Carlo repetitions is ${10^6}$ \cite{2018XinweiFDNOMA,2022XinweiIRSNOMA} and ${a_{an}}$ denotes the power allocation coefficient of AN. The path loss exponent $\alpha $ is set to 2 and the tade-off value of Gauss-Laguerre integration parameter is 300. Suppose that $\emph{K} = 3$ and the variances of complex channel fading coefficients are represented as ${N_{br}} = d_{br}^{ - \alpha }$, ${N_{rk}} = d_{rk}^{ - \alpha }$ and ${N_{re}} = d_{re}^{ - \alpha }$, respectively. Without loss of the generality, the secure performance of RIS-OMA and conventional OMA transmission schemes are considered as benchmarks. Specifically, time division multiple access is adopted for RIS-OMA networks, where each user receives its own signal in one specific time slot and the entire process occupies a total of \emph{K} orthogonal time slots. The AF relaying works in HD mode with a amplification factor of 2. The HD DF relaying equipped with a single antenna consumes two time slots to complete the communication, i.e., one for the BS-DF relaying transmission and the other for the DF relaying-users transmission, respectively. The FD DF relaying includes a pair of transceiver antennas \cite{shi2022secure}. In addition, the secrecy performance of AN-aided NOMA networks is also taken into account, where the FD DF relaying will send AN based on a pseudo-random sequence instead of RIS. Note that the sequence is known to the legitimate user, but remains unknown to the Eves \cite{feng2016robust,zou2016physical}.

\subsection{External Eavesdropping Scenario}
In this subsection, the secrecy outage behaviours of RIS-NOMA networks is illustrated under external eavesdropping scenario.

Fig. 2 plots the SOP versus transmitting SNR with the simulation setup \emph{M} = 16, \emph{P} = 2, \emph{Q} = 8 and $R_1^{EE} = R_2^{EE} = R_3^{EE}$ = 0.04 BPCU in external eavesdropping scenario. The analysis curves of SOP for LUs can be plotted according to (\ref{SOP EE decode k ipSIC}) and (\ref{SOP EE decode k pSIC}), which are consistent with the simulation results. The accuracy of derivation is verified as the asymptotic curves realize convergence according to (\ref{AsySOP EE decode k ipSIC}), (\ref{AsySOP EE decode k pSIC M1}) and (\ref{AsySOP EE decode k pSIC M2}). One can observe that the secrecy performance of the nearest LU (\emph{k} = 3) with pSIC is always superior to that of the distant LUs (\emph{k} = 1, 2). The reason lies in that the near user obtains larger secrecy diversity order, which agrees with the conclusions in \textbf{Remark \ref{Remark2}}. Another observation is that the SOP curves for users under ipSIC converges to an error floor at the high SNR and acquires a zero secrecy diversity order. This phenomenon can also be verified by the insights in \textbf{Remark \ref{Remark1}}. In addition, we can see that the secrecy performance of RIS-NOMA exceeds that of RIS-OMA, AF relaying, HD/FD DF relaying and AN-aided NOMA schemes. This can be understood as follows: 1) RIS-NOMA has the ability to guarantee user fairness more effectively under multi-user cases; 2) RIS-NOMA operating in FD mode offers higher spectrum efficiency in comparison to HD/FD DF relaying; 3) RIS can provide superior security improvements for NOMA networks with lower energy consumption than confusing the Eves by asking the FD DF relaying to emit AN.

\begin{figure}[t!]
    \begin{center}
        \includegraphics[width=2.784in,  height=2.24in]{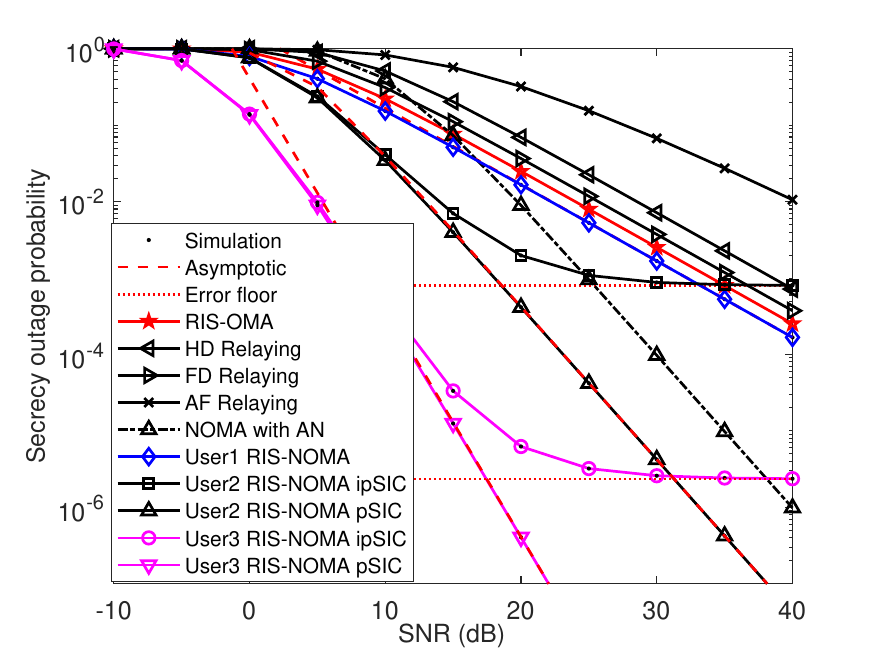}
        \caption{SOP versus transmit SNR under external eavesdropping scenario, with \emph{M} = 16, \emph{P} = 2, \emph{Q} = 8, ${{\rho _e}}$ = 0 dB, $\mathbb{E}\{ {\left| {{h_{ipu}}} \right|^2}\}  = \mathbb{E}\{ {\left| {{h_{ipe}}} \right|^2}\} $  = -20 dB, $R_1^{EE} = R_2^{EE} = R_3^{EE}$ = 0.04 and $R_{OMA}$ = 0.12 BPCU.}
        \label{SOP_EE_diff_SNR}
    \end{center}
\end{figure}

Fig. \ref{SOP_diff_Rate} plots the SOP versus transmitting SNR for all LUs with different target secrecy rates under external eavesdropping scenarios, where \emph{M} = 12, \emph{P} = 2, \emph{Q} = 6 and  $\mathbb{E}\{ {\left| {{h_{ipu}}} \right|^2}\}  = \mathbb{E}\{ {\left| {{h_{ipe}}} \right|^2}\} $  = -10 dB. One can make the following observation from figure that with the increasing of target secrecy rate, the SOP of each LU rises monotonously, which is consistent with traditional NOMA networks. The reason behind this phenomenon is that the higher target security rates raise the threshold of SOP and partial secrecy capacity with small value will be considered as secrecy outage events. Fig. \ref{SOP_diff_M} plots system SOP versus transmitting SNR with various reflecting elements under external eavesdropping scenarios, where \emph{M} = \emph{Q}, \emph{P} = 1, $\mathbb{E}\{ {\left| {{h_{ipu}}} \right|^2}\}  = \mathbb{E}\{ {\left| {{h_{ipe}}} \right|^2}\} $ = -20 dB and $R_1^{EE} = R_2^{EE} = R_3^{EE}$ = 0.04 BPCU. We can see from the figure that RIS-NOMA is capable of achieving enhanced system SOP as the number of reflecting elements gradually grow from 4 to 20. This is because that passive beamformings can obtain a larger freedom of design space by applying more reflecting elements. Besides, the increased number of reflecting elements contribute to make the cascaded communication links more reliable and provide higher channel gains between BS and LUs.

\begin{figure}[t!]
    \begin{center}
        \includegraphics[width=2.784in,  height=2.24in]{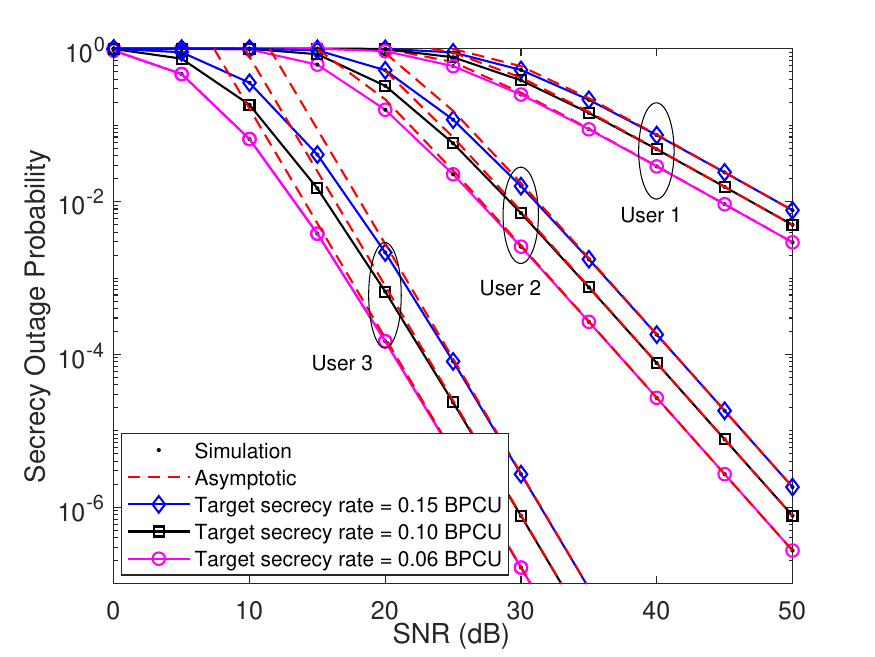}
        \caption{SOP versus transmitting SNR for all LUs with different target secrecy rates under external eavesdropping scenarios, where \emph{M} = 12, \emph{P} = 2, \emph{Q} = 6, ${{\rho _e}}$ = 10 dB, $\mathbb{E}\{ {\left| {{h_{ipu}}} \right|^2}\}  = \mathbb{E}\{ {\left| {{h_{ipe}}} \right|^2}\} $  = -10 dB.}
        \label{SOP_diff_Rate}
    \end{center}
\end{figure}

\begin{figure}[t!]
    \begin{center}
        \includegraphics[width=2.784in,  height=2.24in]{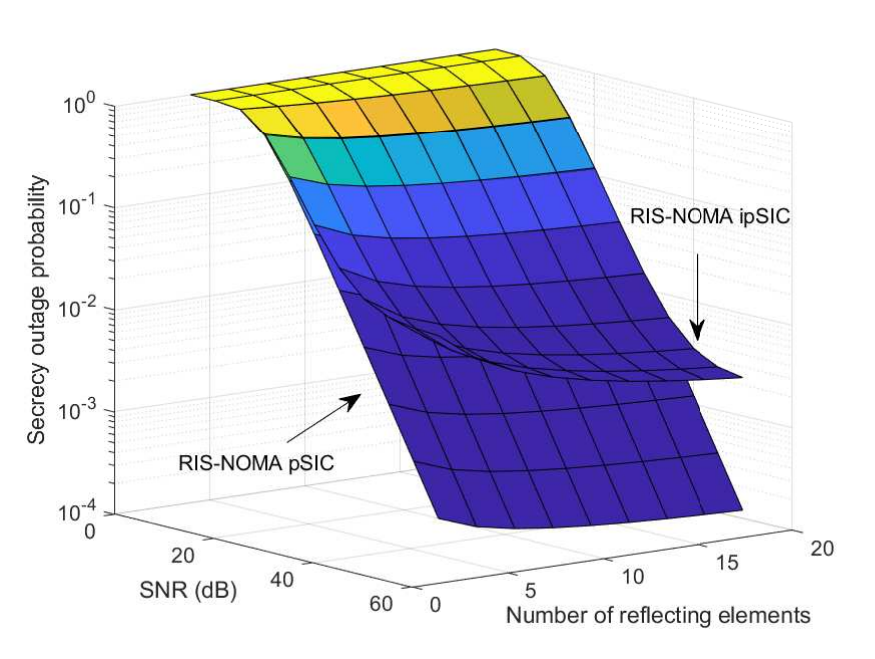}
        \caption{SOP versus transmitting SNR with varying reflecting elements under external eavesdropping scenarios, where \emph{M} = \emph{Q}, \emph{P} = 1, $\mathbb{E}\{ {\left| {{h_{ipu}}} \right|^2}\}  = \mathbb{E}\{ {\left| {{h_{ipe}}} \right|^2}\} $  = -20 dB, ${{\rho _e}}$ = 10 dB and $R_1^{EE} = R_2^{EE} = R_3^{EE}$ = 0.04 BPCU.}
        \label{SOP_diff_M}
    \end{center}
\end{figure}

Fig. 5 plots the SOP versus transmitting SNR with various ${d_{br}}$ and ${d_{rk}}$ under external eavesdropping scenarios, where \emph{M} = 12, \emph{P} = 2, \emph{Q} = 6, $\mathbb{E}\{ {\left| {{h_{ipu}}} \right|^2}\}  = \mathbb{E}\{ {\left| {{h_{ipe}}} \right|^2}\} $ = -20 dB and
$R_1^{EE} = R_2^{EE} = R_3^{EE}$ = 0.04 BPCU. As can be seen from Fig. 5 (a), the security performance of LU is compromised since the distance between BS and RIS increases from 3m to 6m. This behavior is due to the fact that line of sight signals received at RIS become fuzzy due to the serious path fading as the RIS is deployed further away. Similarly, when ${d_{br}}$ and ${d_{re}}$ are both fixed while LUs depart from RIS, the line-of-sight signal is deteriorative and the SOP of users increases severely as indicated in Fig. 5 (b). This is because that the Eve's eavesdropping ability remains unchanged and the superimposed message sent from RIS will suffer more path interference because of the larger transmission distance, which can heavily discount the quality of received signals.

\begin{figure}[htbp]
\centering

\subfigure[Different ${d_{br}}$]{
\begin{minipage}[t]{0.5\linewidth} %linewidth小于0.5
\centering
\includegraphics[width=0.9\textwidth,height=1\textwidth]{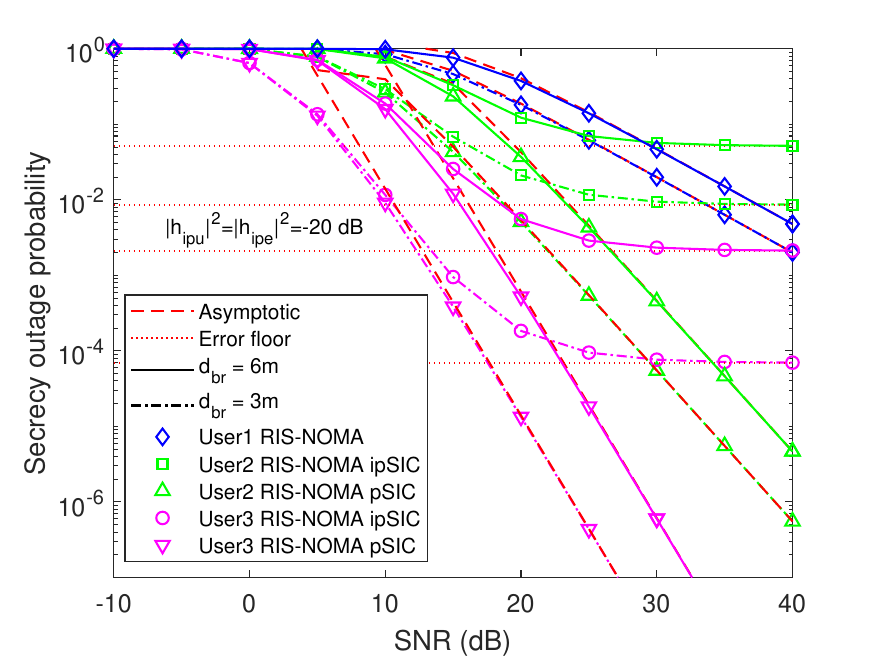}
%\caption{fig1}
\end{minipage}%
}%
\subfigure[Different ${d_{rk}}$]{
\begin{minipage}[t]{0.5\linewidth} %linewidth小于0.5
\centering
\includegraphics[width=0.9\textwidth,height=1\textwidth]{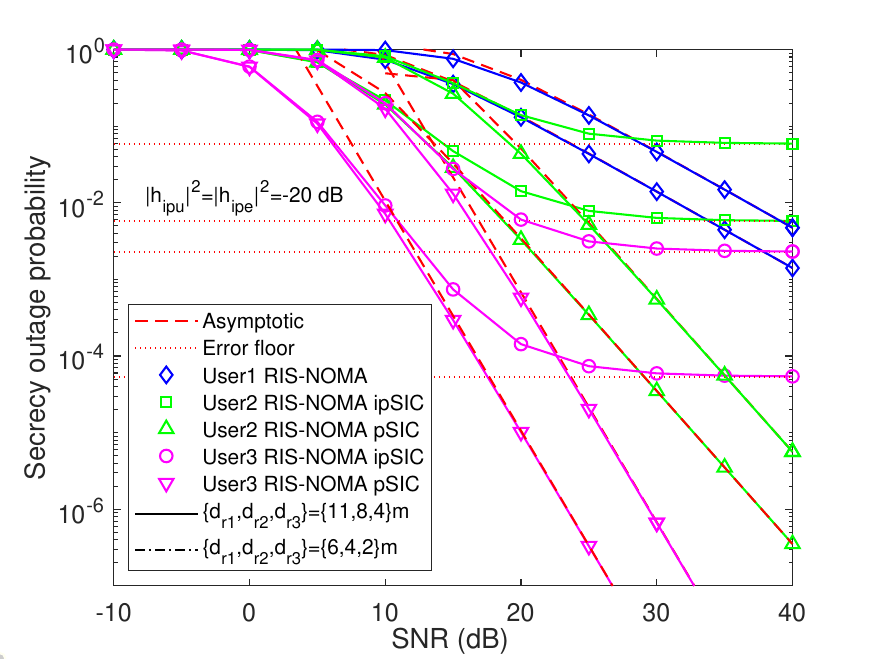}
%\caption{fig2}
\end{minipage}%
}%

\centering
\caption{SOP versus transmitting SNR with various ${d_{br}}$ and ${d_{rk}}$ under external eavesdropping scenarios, where $\mathbb{E}\{ {\left| {{h_{ipu}}} \right|^2}\}  = \mathbb{E}\{ {\left| {{h_{ipe}}} \right|^2}\} $ = -20 dB, ${{\rho _e}}$ = 10 dB, $R_1^{EE} = R_2^{EE} = R_3^{EE}$ = 0.04 BPCU, \emph{M} = 12, \emph{P} = 2 and \emph{Q} = 6.}
\end{figure}

Fig. 6 (a) plots the system SOP versus power allocation with fixed transmitting SNR under external eavesdropping scenarios, where $\rho $ = 10 dB, $\mathbb{E}\{ {\left| {{h_{ipu}}} \right|^2}\}  = \mathbb{E}\{ {\left| {{h_{ipe}}} \right|^2}\} $ = -20 dB, \emph{M} = 16, \emph{P} = 2, \emph{Q} = 8 and ${R_1^{EE} } = {R_2^{EE}}$ = 0.04 BPCU. Note that we consider a pair of users (\emph{K} = 2) in RIS-NOMA networks while the power allocation coefficients for user 1 and user 2 are set as ${a_1} = {a_T}$ and ${a_2} = 1 - {a_T}$, where ${a_T}$ presents the power offset parameter ranging from 0 to 1, i.e., ${a_T} \in \left[ {0,1} \right]$. One can observe from Fig. 6 (a) that the system SOP is sensitive to the variation of power allocation for multiple users and providing the distant LU with a larger power allocation factor is beneficial to the secure performance of RIS-NOMA networks. The reason is that NOMA specializes in allocating more power to the LUs distributed at the edge of cell to ensure they can receive high-quality signals, which enhances the overall ability to resist wiretapping.
\begin{figure}[htbp]
\centering

\subfigure[External eavesdropping scenario]{
\begin{minipage}[t]{0.5\linewidth} %linewidth小于0.5
\centering
\includegraphics[width=0.9\textwidth,height=1\textwidth]{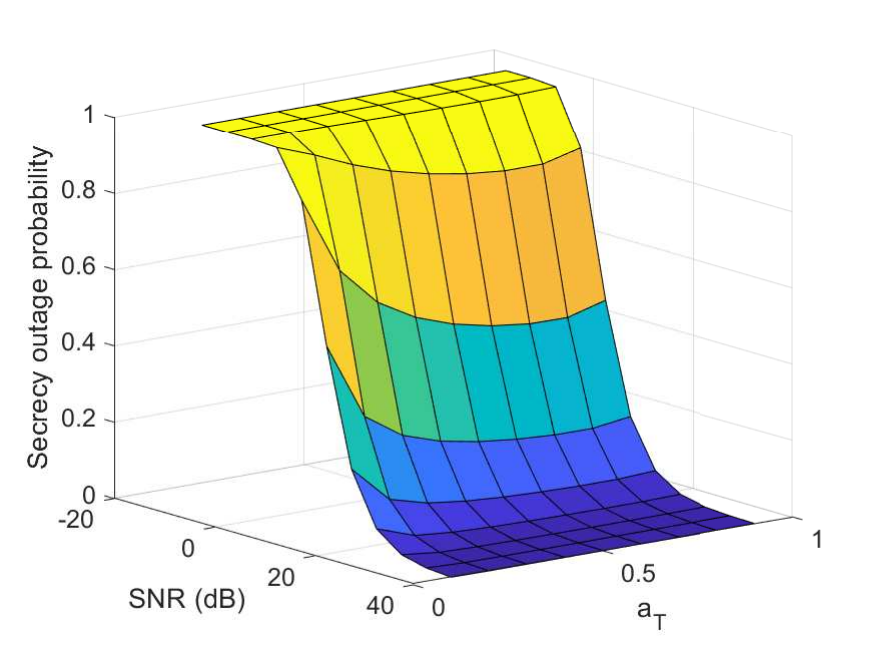}
%\caption{fig1}
\end{minipage}%
}%
\subfigure[Internal eavesdropping scenario]{
\begin{minipage}[t]{0.5\linewidth} %linewidth小于0.5
\centering
\includegraphics[width=0.9\textwidth,height=1\textwidth]{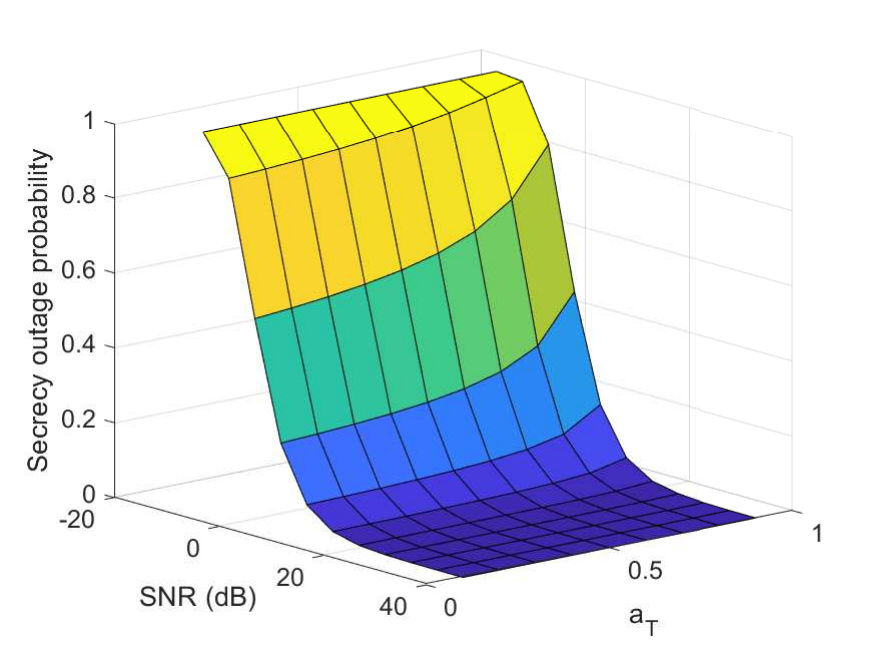}
%\caption{fig2}
\end{minipage}%
}%

\centering
\caption{System SOP versus power allocation with fixed transmitting SNR under different eavesdropping scenarios, where ${{\rho _e}}$ = 10 dB, $\mathbb{E}\{ {\left| {{h_{ipu}}} \right|^2}\}  = \mathbb{E}\{ {\left| {{h_{ipe}}} \right|^2}\} $ = -20 dB, $R_1^\varphi  = R_2^\varphi $ = 0.04 BPCU, $\varphi  \in \left\{ {EE,IE} \right\}$, \emph{M} = 16, \emph{P} = 2 and \emph{Q} = 8.}
\end{figure}

Fig. \ref{SST_EE} plots the secrecy system throughput versus SNR in delay-limited transmission mode under external eavesdropping case, where ${{\rho _e}}$ = 10 dB, $\mathbb{E}\{ {\left| {{h_{ipu}}} \right|^2}\}  = \mathbb{E}\{ {\left| {{h_{ipe}}} \right|^2}\} $ = -20 dB, \emph{M} = \emph{Q} = 16, \emph{P} = 1 and ${R_1^{EE} }$ = 0.08, ${R_2^{EE} }$ = 0.17, ${R_3^{EE} }$ = 0.25 BPCU. The black upper/lower triangular curves represent the secrecy system throughput with ipSIC/pSIC, which can be plotted based on (\ref{SOP EE decode k ipSIC}) and (\ref{SOP EE decode k pSIC}), respectively. It is observed from the figure that the secrecy throughput of RIS-NOMA is notably greater than that of RIS-OMA, AF relaying, FD/HD relaying and AN-aided NOMA schemes. The origin for this behavior can be explained that RIS-NOMA networks has the advantages of high spectral efficiency and reliable channel environment. Another observation is that increasing the value of transmitting SNR, the throughput of various transmission schemes achieves the same convergency value. This is due to the fact that as SNR is on the verge of infinity, the SOPs of LUs become negligible and the secrecy system throughput is dominated by target secrecy rate.

\begin{figure}[t!]
    \begin{center}
        \includegraphics[width=2.784in,  height=2.24in]{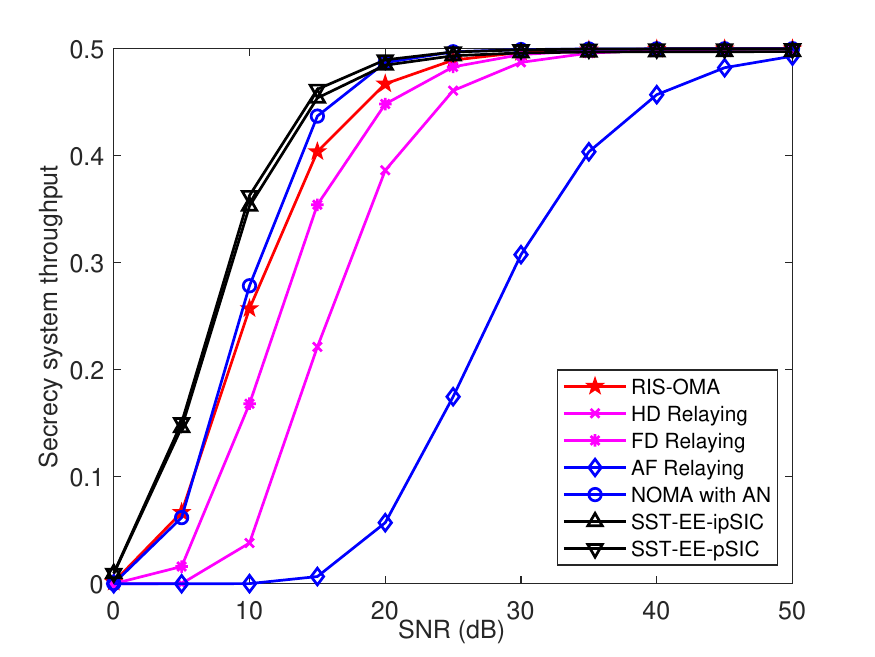}
        \caption{Secrecy system throughput versus transmitting SNR for RIS-NOMA, RIS-OMA and conventional cooperative relaying, where $\mathbb{E}\{ {\left| {{h_{ipu}}} \right|^2}\}  = \mathbb{E}\{ {\left| {{h_{ipe}}} \right|^2}\} $ = -20 dB, ${{\rho _e}}$ = 10 dB, \emph{M} = \emph{Q} = 16, \emph{P} = 1, ${R_1^{EE} }$ = 0.08, ${R_2^{EE} }$ = 0.17 and ${R_3^{EE} }$ = 0.25 BPCU.}
        \label{SST_EE}
    \end{center}
\end{figure}

\subsection{Internal Eavesdropping Scenario}

Fig. \ref{SOP_IE_diff_SNR} plots the SOP versus transmitting SNR in internal eavesdropping scenario, where \emph{M} = 16 and $R_1^{IE} = R_2^{IE} = R_3^{IE}$ = 0.04 BPCU. The analysis curves demonstrated can be acquired by (\ref{SOP IE decode k ipSIC}) and (\ref{SOP IE decode k pSIC}). Additionally, asymptotes of SOP converge in the high SNR region based on (\ref{AsySOP IE decode k ipSIC}), (\ref{AsySOP IE decode k pSIC M1}) and (\ref{AsySOP IE decode k pSIC M2}), which also confirms the correctness of derivation. It can be seen that the internal Eve can also compromise system security and error floors occur for LU 2 and LU 3, which is confirmed in \textbf{Remark \ref{Remark3}}. This is owing to the detrimental influence caused by ipSIC where the previous information is not wiped out completely. Another observation is that LU 3 has a lower outage probability compared with LU 2. This is due to the fact that LU 3 is equipped with higher quality channel conditions, which gives it a greater advantage in anti-eavesdropping. Furthermore, we can also observe that the secrecy outage behaviours of LUs are becoming worse since the setup of on-off control varies from \emph{P} = 2, \emph{Q} = 8 to \emph{P} = \emph{Q} = 4. The reason is that a smaller value of \emph{Q} means fewer elements at RIS are set to 1 (on) for arbitrary ${{\mathbf{v}}_p}$ and the channel gains for LUs are degraded because of the reduction of working elements.

\begin{figure}[t!]
    \begin{center}
        \includegraphics[width=2.784in,  height=2.24in]{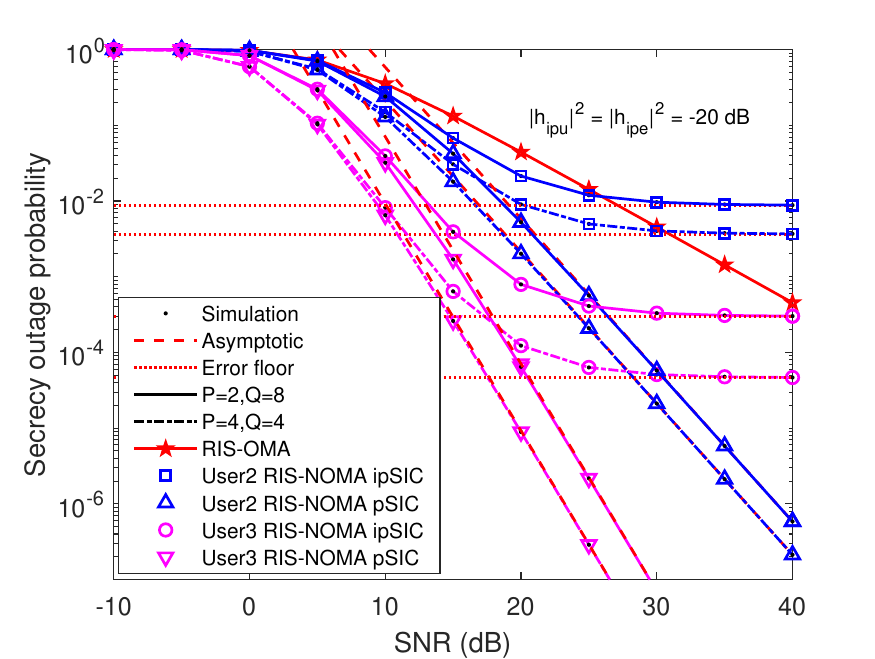}
        \caption{SOP versus transmitting SNR under internal eavesdropping scenario, with \emph{M} = 16, ${{\rho _e}}$ = 10 dB, $R_1^{IE} = R_2^{IE} = R_3^{IE}$ = 0.04 and $R_{OMA}$ = 0.12 BPCU.}
        \label{SOP_IE_diff_SNR}
    \end{center}
\end{figure}

Fig. \ref{SOP_IE_diff_hipu_hipe} plots the SOP versus transmitting SNR as the residual interference varies under internal eavesdropping case, where \emph{M} = 12, \emph{P} = 2, \emph{Q} = 6, $R_1^{IE} = R_2^{IE} = R_3^{IE}$ = 0.04 and $R_{OMA}$ = 0.12 BPCU. The simulation curves for ipSIC can be obtained from (\ref{SOP IE decode k ipSIC}), while the curves presented for pSIC are generated based on (\ref{SOP IE decode k pSIC}). The asymptotic lines further verify the reliability of the derived results. It can be seen from the figure that the residual interference brought by ipSIC will impair the signal decoding process seriously. Moreover, with the increasing of residual interference, the achieved SOP of RIS-NOMA will definitely converge to an inferior error floor. This can be explained that the high levels of ipSIC decrease the received SINR at LUs, which makes it fairly difficult for LUs to recover their own information. In consequence, it is essential to take into account the negative impact of ipSIC on network secure performance in actual communication cases.

\begin{figure}[t!]
    \begin{center}
        \includegraphics[width=2.784in,  height=2.24in]{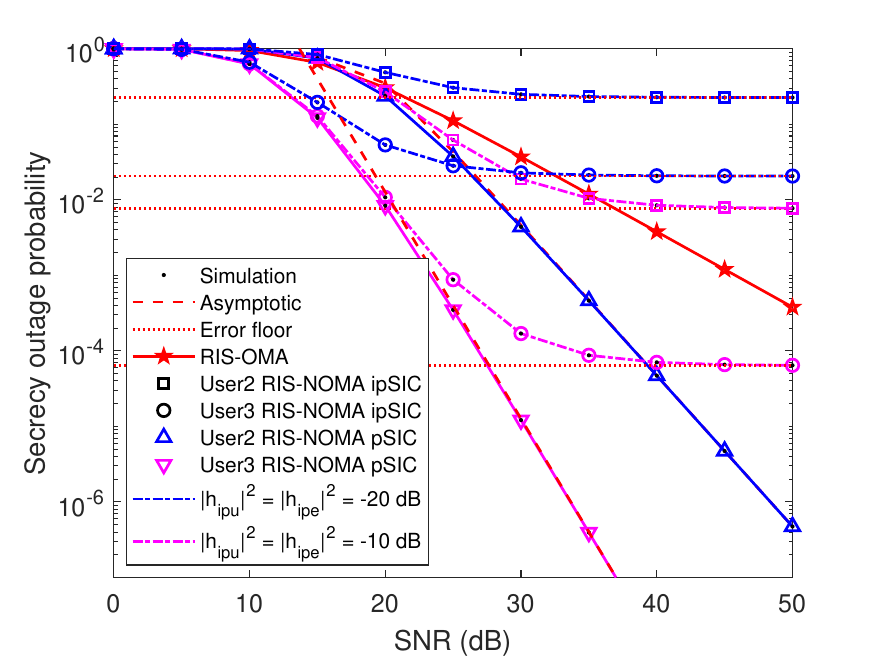}
        \caption{SOP versus transmitting SNR with different residual interference varies under internal eavesdropping case, where \emph{M} = 12, \emph{P} = 2, \emph{Q} = 6, ${{\rho _e}}$ = 5 dB, $R_1^{IE} = R_2^{IE} = R_3^{IE}$ = 0.04 and $R_{OMA}$ = 0.12 BPCU.}
        \label{SOP_IE_diff_hipu_hipe}
    \end{center}
\end{figure}

Fig. 6 (b) plots the system SOP versus power allocation with fixed transmitting SNR under internal eavesdropping scenarios, where $\rho $ = 10 dB, $\mathbb{E}\{ {\left| {{h_{ipu}}} \right|^2}\}  = \mathbb{E}\{ {\left| {{h_{ipe}}} \right|^2}\} $ = -20 dB, \emph{M} = 16, \emph{P} = 2, \emph{Q} = 8 and ${R_1^{IE} } = {R_2^{IE}}$ = 0.04 BPCU. It can be seen from this figure that as the value of ${a_T}$ increases, the secrecy outage behaviours of RIS-NOMA networks become worse seriously, which is opposite to the observation in Fig. 6 (a). The reason is that user 1 with poor channel conditions is regarded as an internal Eve, and allocating more power to user 1 can inevitably strengthen its eavesdropping ability while weakening the received signal quality of LUs, thus reducing the system SOP.

\begin{figure}[t!]
    \begin{center}
        \includegraphics[width=2.784in,  height=2.24in]{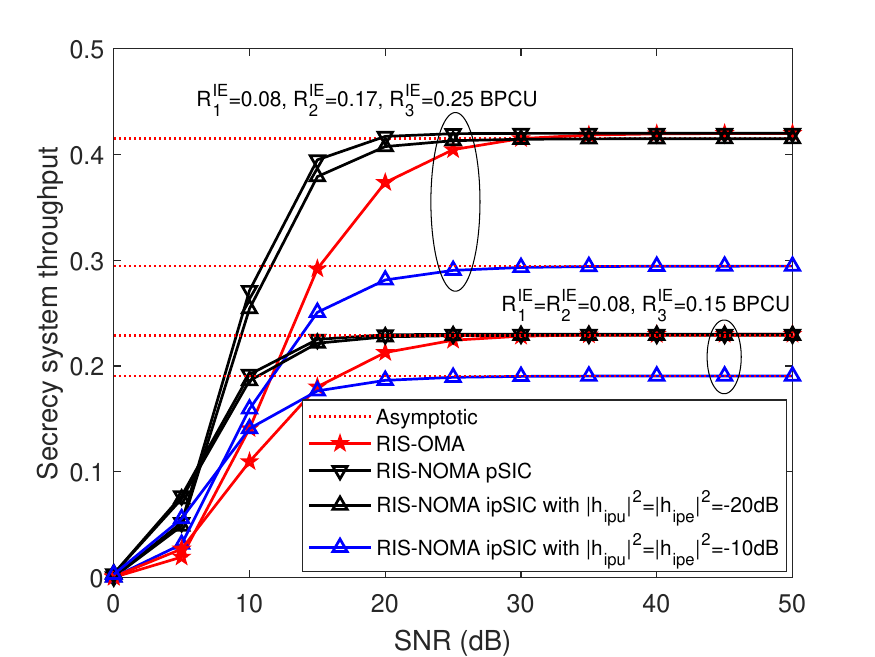}
        \caption{Secrecy system throughput versus transmitting SNR for RIS-NOMA under internal eavesdropping case, where ${{\rho _e}}$ = 10 dB and \emph{M} = 16.}
        \label{SST_IE}
    \end{center}
\end{figure}

Fig. \ref{SST_IE} plots the secrecy system throughput versus SNR in delay-limited transmission mode under internal eavesdropping case, where ${{\rho _e}}$ = 10 dB, \emph{M} = 16, \emph{P} = 2 and \emph{Q} = 8. The analysis curves of secrecy system throughput for RIS-NOMA with ipSIC/pSIC are plotted according to (\ref{SOP IE decode k ipSIC}) and (\ref{SOP IE decode k pSIC}), respectively. As can be observed from the figure, with the increase of residual interference, the secrecy system throughput of RIS-NOMA is distinctly reduced under the circumstances of ipSIC. This is because that the deterioration of imperfect cancellation process will significantly weaken the received SINR at LUs and thus raise the SOP of users, which can directly impair the secrecy throughput performance referring to (\ref{SST define}). Another observation is that the improvement of secrecy system throughput is achieved by increasing the target secrecy rate, while the secrecy outage behaviors are becoming worse when larger target secrecy rate is applied in RIS-NOMA networks as illustrated in Fig. \ref{SOP_diff_Rate}. Therefore, there exists a tradeoff to balance the performance of SOP and secrecy system throughput.

\section{Conclusion}\label{Conclusion}\label{SectionV}
In this paper, the secure communications of RIS-NOMA networks were investigated with on-off control.
New approximate and asymptotic expressions of SOP for the \emph{k}-th LU with ipSIC/pSIC were derived in RIS-NOMA networks. On this basis, the secrecy diversity order of the \emph{k}-th LU was acquired which can be determined by the residential interference and channel ordering. Referring to the analytical results, it was shown that the RIS-NOMA networks can obtain a superior secrecy outage behavior compared with RIS-OMA, AF relaying and HD/FD DF relaying. Furthermore, the expressions of secrecy system throughput in delay-limited transmission mode were derived. The results showed that the throughput performance of RIS-NOMA is much greater than that of conventional cooperative communication schemes.

\appendices
\section*{Appendix~A: Proof of Lemma \ref{Lemma1}} \label{AppendixA}
\renewcommand{\theequation}{A.\arabic{equation}}
\setcounter{equation}{0}
According to the principle of on-off control, the cascaded channel ${{\mathbf{h}}_{_{rk}}^H{\mathbf{\Theta }}{{\mathbf{h}}_{br}}}$ can be defined as ${H_k} \triangleq {\mathbf{h}}_{rk}^H{\mathbf{\Theta }}{{\mathbf{h}}_{br}} = {\mathbf{v}}_p^H{{\mathbf{D}}_{rk}}{{\mathbf{h}}_{br}}$. Hence, the SINR for the \emph{k}-th LU to decode the \emph{g}-th LU's information with ipSIC is shown as  ${{\tilde \gamma }_{k \to g}} =
\mathop {}\nolimits^{} \mathop {}\nolimits^{} \frac{{\rho {{\left| {{{{H}}_{k}}} \right|}^2}{a_g}}}{{\rho {{\left| {{{{H}}_{k}}} \right|}^2}{{\nu _g}} + \varpi \rho {{\left| {{h_{ipu}}} \right|}^2} + 1}}$. Hence, the CDF of ${{\tilde \gamma }_{k \to g}}$ is given by
\begin{align}\label{CDF of k decode g ipSIC temp}
  F_{{{\tilde \gamma }_{k \to g}}}^{ipSIC}\left( x \right) &= P\left( {{{\tilde \gamma }_{k \to g}} < x} \right) \notag \\
   &= P\left( {\frac{{\rho {{\left| {{{{H}}_{k}}} \right|}^2}{a_g}}}{{\rho {{\left| {{{{H}}_{k}}} \right|}^2}{{\nu _g}} + \varpi \rho {{\left| {{h_{ipu}}} \right|}^2} + 1}} < x} \right) \notag \\
   &= \int_0^\infty  {{f_{{{\left| {{h_{ipu}}} \right|}^2}}}} \left( y \right){F_{{{\left| {{{{H}}_{k}}} \right|}^2}}}\left[ {\frac{{x\left( {\varpi \rho y + 1} \right)}}{{\left( {{a_g} - {{\nu _g}}x} \right)\rho }}} \right]dy,
\end{align}
where ${f_{{{\left| {{h_{ipu}}} \right|}^2}}}\left( y \right) = \frac{1}{{{N_{ipu}}}}{e^{ - \frac{y}{{{N_{ipu}}}}}}$. The PDF of the cascaded channels from BS to RIS and to LUs can be expressed as follows \cite{liu2014outage}
\begin{small}
\begin{align}\label{PDF of cascaded channel NotORD}
{f_{{{\left| {{{{H}}_{k}}} \right|}^2}}}\left( z \right) = \frac{{2{z^{\frac{{Q - 1}}{2}}}}}{{\Gamma \left( Q \right){{\left( {\sqrt {{N_{br}}{N_{rk}}} } \right)}^{Q + 1}}}}{K_{Q - 1}}\left( {2\sqrt {\frac{z}{{{N_{br}}{N_{rk}}}}} } \right).
\end{align}
\end{small}
Applying integration operation and some simple manipulations to (\ref{PDF of cascaded channel NotORD}), it can be rewritten as
\begin{align}\label{CDF of cascaded channel NotORD temp}
{F_{{{\left| {{{{H}}_{k}}} \right|}^2}}}\left( z \right) =& \frac{4}{{\Gamma \left( Q \right)}}{\left( {\frac{z}{{{N_{br}}{N_{rk}}}}} \right)^{\frac{{Q + 1}}{2}}}\notag \\ &\times \int_0^1 {{y^Q}} {K_{Q - 1}}\left( {2\sqrt {\frac{z}{{{N_{br}}{N_{rk}}}}} y} \right)dy.
\end{align}
According to \cite[Eq. (6.561.8)]{gradvstejn2000table}, ${F_{{{\left| {{{{H}}_{k}}} \right|}^2}}}\left( z \right)$ can be further rewritten as
\begin{small}
\begin{align}\label{CDF of cascaded channel NotORD}
{F_{{{\left| {{{{H}}_{k}}} \right|}^2}}}\left( z \right) = 1 - \frac{2}{{\Gamma \left( Q \right)}}{\left( {\frac{z}{{{N_{br}}{N_{rk}}}}} \right)^{\frac{Q}{2}}}{K_Q}\left( {2\sqrt {\frac{z}{{{N_{br}}{N_{rk}}}}} } \right).
\end{align}
\end{small}
Considering ${\left| {{\mathbf{h}}_{_{r1}}^H{\mathbf{\Theta }}{{\mathbf{h}}_{br}}} \right|^2} \leqslant  \cdots \leqslant{\left|{{\mathbf{h}}_{_{rk}}^H{\mathbf{\Theta }}{{\mathbf{h}}_{br}}} \right|^2} \leqslant  \cdots  \leqslant {\left| {{\mathbf{h}}_{_{rK}}^H{\mathbf{\Theta }}{{\mathbf{h}}_{br}}} \right|^2}$, the CDF of cascaded channel with order is given by \cite{david2004order,2020orderZhiguo}
\begin{small}
\begin{align}\label{CDF of cascaded channel ORD_temp}
F_{{{\left| {{{{H}}_{k}}} \right|}^2}}^{ORD}\left( z \right) &= \kappa \sum\limits_{l = 0}^{K - k} {{{
  {K - k} \choose
  l }}\frac{{{{\left( { - 1} \right)}^l}}}{{k + l}}}\notag \\  &\times \left[ {1 - \frac{2}{{\Gamma \left( Q \right)}}{{\left( {\frac{z}{{{N_{br}}{N_{rk}}}}} \right)}^{\frac{Q}{2}}}{K_Q}\left( {2\sqrt {\frac{z}{{{N_{br}}{N_{rk}}}}} } \right)} \right]^{k + l}.
\end{align}
\end{small}

Upon substituting (\ref{CDF of cascaded channel ORD_temp}) into (\ref{CDF of k decode g ipSIC temp}), we obtain
\begin{small}
\begin{align}\label{CDF of cascaded channel ORD}
F_{{{\tilde \gamma }_{k \to g}}}^{ipSIC}\left( x \right) =& \frac{{\kappa }}{{{N_{ipu}}}}\sum\limits_{l = 0}^{K - k} {{{
  {K - k} \choose
  l }}\frac{{{{\left( { - 1} \right)}^l}}}{{k + l}}}  \notag \\
&\times \int_0^\infty  {{e^{ - \frac{y}{{{N_{ipu}}}}}}} \left\{ {1 - \frac{2}{{\Gamma \left( Q \right)}}{{\left[ {\frac{{x\left( {\varpi \rho y + 1} \right)}}{{\left( {{a_g} - {{{\nu _g}} }x} \right){\zeta _2}}}} \right]}^{\frac{Q}{2}}}} \right. \notag \\
&\times {\left. {{K_Q}\left[ {2\sqrt {\frac{{x\left( {\varpi \rho y + 1} \right)}}{{\left( {{a_g} - {{\nu _g}} x} \right){\zeta _2}}}} } \right]} \right\}^{k + l}}dy.
\end{align}
\end{small}
Replacing ${y \mathord{\left/
 {\vphantom {y {{N_{ipu}}}}} \right.
 \kern-\nulldelimiterspace} {{N_{ipu}}}}$ with \emph{t}, (\ref{CDF of cascaded channel ORD}) can be rewritten as
\begin{small}
\begin{align}\label{CDF of cascaded channel ORD 2}
F_{{{\tilde \gamma }_{k \to g}}}^{ipSIC}\left( x \right) =& {\kappa }\sum\limits_{l = 0}^{K - k} {{{
  {K - k} \choose
  l }}\frac{{{{\left( { - 1} \right)}^l}}}{{k + l}}}  \notag \\
&\times \int_0^\infty  {{e^{ - {t}}}} \left\{ {1 - \frac{2}{{\Gamma \left( Q \right)}}{{\left[ {\frac{{x\left( {\varpi \rho {{N_{ipu}}t} + 1} \right)}}{{\left( {{a_g} - {{{\nu _g}} }x} \right){\zeta _2}}}} \right]}^{\frac{Q}{2}}}} \right. \notag \\
&\times {\left. {{K_Q}\left[ {2\sqrt {\frac{{x\left( {\varpi \rho {{N_{ipu}}t} + 1} \right)}}{{\left( {{a_g} - {{\nu _g}} x} \right){\zeta _2}}}} } \right]} \right\}^{k + l}}dt.
\end{align}
\end{small}By utilizing Gauss-Laguerre integration and some mathematical operations, we can obtain (\ref{CDF SINR k decode g ipSIC}) and the proof is completed.
\section*{Appendix~B: Proof of Lemma \ref{Lemma2} } \label{AppendixB}
\renewcommand{\theequation}{B.\arabic{equation}}
\setcounter{equation}{0}
The PDF of wiretap channels can be expressed as follows based on (\ref{PDF of cascaded channel NotORD})
\begin{small}
\begin{align}\label{PDF of cascaded wiretap channel ORD}
{f_{{{\left| {{{{H}}_{e}}} \right|}^2}}}\left( v \right) = \frac{{2{v^{\frac{{Q - 1}}{2}}}}}{{\Gamma \left( Q \right){{\left( {\sqrt {{N_{br}}{N_{re}}} } \right)}^{Q + 1}}}}{K_{Q - 1}}\left( {2\sqrt {\frac{v}{{{N_{br}}{N_{re}}}}} } \right).
\end{align}
\end{small}
As a consequence, the expectation of ${{\left| {{{{H}}_{e}}} \right|}^2}$ is shown as
\begin{small}
\begin{align}\label{Exp of cascaded EEchannel}
\mathbb{E}\left( {{{\left| {{{{H}}_{e}}} \right|}^2}} \right) =& \int_0^\infty  {v{f_{{{\left| {{{{H}}_{e}}} \right|}^2}}}\left( v \right)dv}\notag \\  =& \frac{2}{{\Gamma \left( Q \right){{\left( {\sqrt {{N_{br}}{N_{re}}} } \right)}^{Q + 1}}}}\notag \\ &\times\int_0^\infty  {{v^{\frac{{Q + 1}}{2}}}{K_{Q - 1}}\left( {2\sqrt {\frac{v}{{{N_{br}}{N_{re}}}}} } \right)dv} .
\end{align}
\end{small}
Assuming $v = {x^2}$, (\ref{Exp of cascaded EEchannel}) can be recast as
\begin{small}
\begin{align}\label{Exp of cascaded EEchannel 2}
\mathbb{E}\left( {{{\left| {{{{H}}_{e}}} \right|}^2}} \right) =& \frac{4}{{\Gamma \left( Q \right){{\left( {\sqrt {{N_{br}}{N_{re}}} } \right)}^{Q + 1}}}}\notag \\ &\times\int_0^\infty  {{x^{{{Q + 2}}}}{K_{Q - 1}}\left( {2{x}\sqrt {\frac{1}{{{N_{br}}{N_{re}}}}} } \right)dx}.
\end{align}
\end{small}
With the assistance of \cite[Eq. (6.561.16)]{gradvstejn2000table}, we can acquire $\mathbb{E}\left( {{{\left| {{{{H}}_{e}}} \right|}^2}} \right) = Q{N_{br}}{N_{re}}$. According to (\ref{SINR EE decode k v2}) and (\ref{SOP k temp}), $P_{out,k}^{ipSIC,EE}$ is given by
\begin{small}
\begin{align}\label{SOP of EE decode k AppendixB}
\begin{gathered}
  P_{out,k}^{ipSIC,EE}\left( {R_{_k}^{EE}} \right) = P\left[ {{{\tilde \gamma }_k} < {2^{R_k^{EE}}}\left( {1 + {{\tilde \gamma }_{EE \to k}}} \right) - 1} \right] \hfill \\
   \approx P\left\{ {{{\tilde \gamma }_k} < {2^{R_k^{EE}}}\left[ {1 + \frac{{\mathbb{E}\left( {{{\left| {{{{H}}_{e}}} \right|}^2}} \right){a_k}{\rho _e}}}{{\mathbb{E}\left( {{{\left| {{{{H}}_{e}}} \right|}^2}} \right){\rho _e}{{{\nu _k}} } + \varpi {\rho_e} {{\left| {{h_{ipe}}} \right|}^2} + 1}}} \right] - 1} \right\} \hfill \\
   = \int_0^\infty  {{f_{{{\left| {{h_{ipe}}} \right|}^2}}}} \left( y \right){F_{{{\left| {{H_{k}}} \right|}^2}}}\left[ {{2^{R_k^{EE}}}\left( {1 + \frac{{{\chi _1}}}{{{\chi _2} + \varpi {\rho_e} y}}} \right) - 1} \right]dy,
\end{gathered}
\end{align}
\end{small}where ${f_{{{\left| {{h_{ipe}}} \right|}^2}}}\left( y \right) = \frac{y}{{{N_{ipe}}}}{e^{ - \frac{1}{{{N_{ipe}}}}}}$, ${\chi _1} = Q{N_{br}}{N_{re}}{a_k}{\rho _e}$ and ${\chi _2} = Q{N_{br}}{N_{re}} {{\nu _k}}  + 1$. Similar to (\ref{CDF of cascaded channel ORD 2}), after replacing ${y \mathord{\left/
 {\vphantom {y {{N_{ipe}}}}} \right.
 \kern-\nulldelimiterspace} {{N_{ipe}}}}$ with \emph{t}, (\ref{SOP EE decode k ipSIC}) can be derived with the aid of Gauss-Laguerre integration. The proof is completed.

\bibliographystyle{IEEEtran}
\bibliography{VT-2022-02532_bib}

\end{document}